\documentclass[runningheads,a4paper]{llncs}
\usepackage{amssymb,amsmath,pifont,footnote}
\usepackage{latexsym,pdfpages}
\usepackage{array,verbatim}
\usepackage{stmaryrd,ifsym}
\usepackage{algorithmic,algorithm}
\usepackage{cite}
\usepackage{fullpage}
\usepackage{pstricks}
\usepackage{gastex}
\newtheorem{thm}{Theorem} 

\newtheorem{lem}{Lemma}

\newtheorem{rem}{Remark}
\newtheorem{prob}{Problem}

\newtheorem{defn}{Definition}
\newtheorem{examp}{Example}

\newcommand{\pfbox}{}

\newcommand{\Q}{\mathbb{Q}}
\newcommand{\R}{\mathbb{R}}

\newcommand{\Set}[1]{\{ #1 \}}
\newcommand{\Range}[2]{#1 ,\dots, #2}
\newcommand{\RangeSet}[2]{\Set{ \Range{#1}{#2}}}

\newcommand{\MPInfLeq}[1]{\ensuremath{\textrm{MeanPayoffInf}^{\leq}(#1)}}

\newcommand{\AttrOne}[1]{\ensuremath{Attr_1(#1)}}

\newcommand{\HQ}{\mathit{H10(\Q)}}

\newcommand{\Heading}[1]{\smallskip\noindent{\bf{#1}}}

\newcommand{\BeginProof}{\vspace{-0.25cm}\begin{proof}}

\newcommand{\Comment}[1]{}
\newcommand{\Appendix}[1]{}

\newcommand{\MAX}{\max}
\newcommand{\MIN}{\min}
\newcommand{\SUM}{\operatorname{sum}}
\newcommand{\OP}{\operatorname{op}}
\newcommand{\Avg}{\mathit{Avg}}
\newcommand{\LimAvg}{\mathit{LimAvg}}
\newcommand{\LimInfAvg}{\mathit{LimInfAvg}}
\newcommand{\LimSupAvg}{\mathit{LimSupAvg}}

\newcommand{\VEC}[1]{\ensuremath{\overline{#1}}}

\newcommand{\Nat}{\ensuremath{\mathbb{N}}}

\newcommand{\CONV}{\mathit{CONVEX}}

\newcommand{\Val}{\ensuremath{\mathit{Val}}}
\newcommand{\ValP}{\ensuremath{\mathit{Val}}}

\newcommand{\FM}{\mathcal{FM}}
\newcommand{\FiniteOne}{\mathcal{FM}_1}

\newcommand{\Strategies}{\mathcal{S}}

\newcommand{\StrategiesTwo}{\Strategies_2}
\newcommand{\ValueRegion}{\mathit{VR}}

\newcommand{\Obj}{\mathit{obj}}

\newcommand{\VectorSetV}{\mathcal{V}}

\newcommand{\Simplex}{\mathcal{S}}
\newcommand{\IntSimplex}{\mathcal{SI}}

\newcommand{\RatSimplex}{\mathcal{QSI}}
\newcommand{\WinOne}{\ensuremath{\mathit{Win}_1}}

\newcommand{\MaxFreeConst}{\operatorname{MFC}}

\begin{document}

\pagestyle{plain}

\newcommand{\AppendixString}{\cite{mpexpressionCorr}}
\newcommand{\AppendixContent}[1]{}

\title{Finite-Memory Strategy Synthesis for Robust Multidimensional Mean-Payoff Objectives}
\author{Yaron Velner}
\institute{The Blavatnik School of Computer Science, Tel Aviv University, Israel}
\titlerunning{Strategy Synthesis for Mean-Payoff Expressions}
\maketitle

\begin{abstract}
Two-player games on graphs provide the mathematical foundation for the study of reactive systems.
In the quantitative framework, an objective assigns a value to every play, and the goal of player 1 is to minimize the value of the objective.
In this framework, there are two relevant synthesis problems to consider: the quantitative analysis problem is to compute the minimal (or infimum) value that player 1 can assure, and the boolean analysis problem asks whether player 1 can assure that the value of the objective is at most $\nu$ (for a given threshold $\nu$).
Mean-payoff expression games are played on a multidimensional weighted graph.
An atomic mean-payoff expression objective is the mean-payoff value (the long-run average weight) of a certain dimension, 
and the class of mean-payoff expressions is the closure of atomic mean-payoff expressions under the algebraic operations of $\MAX,\MIN$, numerical complement and $\SUM$.
In this work, we study for the first time the strategy synthesis problems for games with robust quantitative objectives, namely, games with mean-payoff expression objectives.
While in general, optimal strategies for these games require infinite-memory, in synthesis we are typically interested in the construction 
of a finite-state system.
Hence, we consider games in which player 1 is restricted to finite-memory strategies, and our main contribution is as follows.
We prove that for mean-payoff expressions, the quantitative analysis problem is computable, and the boolean analysis problem is inter-reducible with Hilbert's tenth problem over rationals --- a fundamental long-standing open problem in computer science and mathematics.
\end{abstract}

\sloppy
\section{Introduction}
In the classical framework of boolean formal verification, a program may only violate or satisfy a given specification, and in the framework of synthesis, the task is to automatically construct \textbf{a} program that satisfies the specification.
The boolean framework does not discriminate between programs that satisfy a given specification, and consequently, it may produce (or verify) unreasonable implementations.
 
In the recent years, there is an emerging line of research that aims to measure the quality of a program with quantitative metrics, e.g.,~\cite{Weighted-Automaton,Alur:2009:ODM:1532848.1532880,Closure,BokerCHK11,DrosteR06,BrazdilCKN12}.
The quantitative verification problem asks how well a program satisfies a given specification, and the synthesis task is to construct the optimal program with respect to a specification.

Quantitative verification and synthesis problems are modelled by infinite-duration games over weighted graphs.
In these games, the set of vertices is partitioned into player-1 and player-2  vertices;
initially, a pebble is placed on an initial vertex, and in every round, the player who owns the vertex that the pebble resides in, advances the pebble to an adjacent vertex.
This process is repeated forever and gives rise to a \emph{play} that induces an infinite sequence of weights (or weight vectors), and a quantitative objective
assigns a value to every play (or equivalently to every infinite sequence of weights).

The classical work on these games only considered games with single objectives, such as minimizing the long run average weight, or minimize the sum of weights. In order to have robust quantitative specifications, it is necessary to investigate games on graphs with multiple (and possibly conflicting) objectives.
Typically, multiple objectives are modeled by multidimensional weight functions (e.g.,~\cite{BrazdilBCFK11,BrazdilCKN12,ChatterjeeRR12,Alur:2009:ODM:1532848.1532880}), and the outcome of a \emph{play} is a vector of values.
In the boolean setting, the goal of player 1 is to satisfy a boolean condition on the values (with respect to a threshold vector).
For example, player 1 needs to assure that the average response time ($\mathit{rt}$) of an arbiter is at most $2.4$ and that the average energy consumption ($\mathit{ec}$) is below $7$.
In the quantitative setting, the outcome of a play is a unique (real) value, and the goal of player 1 is to minimize the value of the play.
A multiple objective specification is modelled by algebraic operations on single objectives.
In the example above, we define the quantitative objective $\MAX(\mathit{rt} - 2.4,\mathit{ec} - 7)$, and a non-positive value to the quantitative objective implies that the boolean objective is satisfied.
In the general case, an objective is determined either by the projection of the weight function to one dimension, or it is formed by algebraic operation on two (or more) objectives.
In the literature, the common and natural algebraic operations are $\MIN$,$\MAX$, \emph{numerical complement} (multiplication by $-1$) and $\SUM$.
We note that when the goal is to minimize the value of the objective, then the first three operations generalize the boolean disjunction, conjunction and negation.
A class of quantitative objectives is \emph{robust} if it is closed under the four algebraic operations.
So far, the only known class of robust quantitative objectives that has an effective algorithm for the model checking problem (that is, for solving one-player games) is the class of \emph{mean-payoff expressions}~\cite{mean-payoff-Automaton-Expressions}, which is the closure of one-dimensional mean-payoff (long-run average of the weights) objectives to the four algebraic operation.
For example, for an infinite sequence of vectors $a = a_1,a_2,\dots \in (\R^3)^\omega$ the objective \[E(a) = \LimAvg_1(a) + \MIN(\LimAvg_2(a),-\LimAvg_3(a))\] is a mean-payoff expression (where $\LimAvg_i$ is the long-run average of dimension $i$) and $E((1,2,6)^\omega) = 1 + \min(2,-6) = -5$.

In the quantitative setting, there are two relevant synthesis problems: (i)~the \emph{quantitative analysis problem} is to compute the optimal (infimum) value that a player-1 strategy can assure; and (ii)~the \emph{boolean analysis problem} is to determine whether player 1 can assure a value of at most $\nu$ to the objective (for a given $\nu$).
From the perspective of synthesis, these problems are most important when player 1 is restricted to finite-memory strategies (in Example~\ref{examp:InfiniteMem} we show that infinite-memory strategies may yield a better value for player~1, hence the restriction to finite-memory strategies may affect the analysis of the synthesis problem).

For mean-payoff expressions, optimal finite-memory strategies may not always exist. Hence, the quantitative analysis problem is to compute the greatest lower bound on the minimal value that player 1 can assure.
We note that since all model checking problems (i.e., the quantitative generalization of the emptiness, universality and language inclusion) are decidable for mean-payoff expression, then the computability of the quantitative analysis will give us an effective algorithm to synthesize $\epsilon$-optimal finite-memory strategies, and if the boolean analysis problem were decidable, then we would have an algorithm that construct the corresponding player-1 strategy.

\Heading{Our contribution.}
In this paper, we consider for the first time the synthesis problem for a robust class of quantitative objectives.
We prove computability for the quantitative synthesis problem, and we show that the boolean analysis problem is inter-reducible with Hilbert's tenth problem over rationals ($\HQ$), which is a fundamental long-standing open question in computer science and mathematics.
We show that the problem is inter-reducible with $\HQ$ even when both players are restricted to finite-memory strategies, and we show that there is a fragment of mean-payoff expressions that is $\HQ$-hard when one or both players are restricted to finite-memory strategies, but decidable when both players may use infinite-memory strategies.

Our main technical contribution is the introduction of a general scheme that lifts a one-player game solution (equivalently, a model checking algorithm) to a solution for a two-player game (when player 1 is restricted to finite-memory strategies).
The scheme works for a large class of quantitative objectives that have certain properties (which we define in Subsection~\ref{subsec:ConvObj}).

\Heading{Related work.}
The class of mean-payoff expressions was introduced in~\cite{mean-payoff-Automaton-Expressions}, and the decidability of the model checking problems (which correspond to one-player games) was established. A simpler and more efficient algorithm for mean-payoff expression games was given in~\cite{YaronMPExp}.
Mean-payoff games on multidimensional graphs were first studied in~\cite{ChatterjeeDHR10}.
In these games the objective of player 1 was to satisfy a conjunctive condition (in the terms of this paper, the objective was a maximum of multiple one-dimensional objectives).
In~\cite{VelnerR11}, decidability for an objective that is formed by the $\MIN$ and $\MAX$ operators was established.
But the proof can not be extended to include the numerical complement operator, and it does not scale for the case that player 1 is restricted to finite-memory strategies.

\Heading{Structure of the paper.}
In the next section we give the basic definitions for quantitative games and we 
define a class of quantitative objectives that have special properties.
In Sections~\ref{sect:SPANGame} and~\ref{sect:Generic} we give a generic solution for the synthesis problem of quantitative objectives that satisfies the special properties (and an overview of the solution is given in Subsection~\ref{sect:Overview}).
In Section~\ref{sect:OnePlayerMPExpProp} we show that mean-payoff expressions satisfy the special properties and the main results of the paper follow.
Some of the proofs were omitted from the main paper, and the full proofs are given in the appendix.
\section{Games with Quantitative Objectives}\label{sect:Prem}
In this section we give the formal definitions for quantitative objectives and games on graphs with quantitative objectives (Subsection~\ref{subsec:Defs}).
We define four special properties of quantitative objectives (Subsection~\ref{subsec:ConvObj}), and we give an informal overview for the two-player game solution of games with quantitative objectives that satisfy the special properties (Subsection~\ref{sect:Overview}). 
\subsection{Quantitative games on graphs}\label{subsec:Defs}

\Heading{Quantitative objectives.}
In this paper we consider directed finite graphs with a $k$-dimensional weight function that assigns a vector of rationals to each edge.
A \emph{quantitative objective} is a function that assigns a value to every infinite sequence of weight vectors. Formally an objective is a function $\Obj:(\R^k)^\omega \to \R$.
A simple example for quantitative objective is to consider a one-dimensional weight function and an objective that assigns to each infinite path the maximal weight that occurs infinitely often in the path.
An objective $\Obj : (\R^k)^\omega\to\R$ is called \emph{prefix-independent} if for every $a_1\in (\R^k)^*$ and $a_2\in (\R^k)^\omega$ it holds that $\Obj(a_1 a_2) = \Obj(a_2)$.

\Heading{Algebraic operations over quantitative objectives.}
The quantitative counterpart of the boolean operations of disjunction, conjunction and negation are the $\MAX$,$\MIN$ and numerical complement operators (numerical complement is multiplication by $-1$). The $\SUM$ operator, which does not have a boolean counterpart, is also very natural operator in the framework of quantitative objectives.
For two quantitative objectives $\Obj_1$ and $\Obj_2$, the quantitative objective $\OP(\Obj_1,\Obj_2)$ (for $\OP\in\{\MIN,\MAX,\SUM\}$) assigns to every infinite sequence of weights $\ell\in (\R^k)^\omega$ the value $\OP(\Obj_1(\ell),\Obj_2(\ell))$, and the numerical complement of $\Obj_1$ assigns the value of $-\Obj_1(\ell)$.

\Heading{Robust quantitative objectives}
A class of quantitative objectives $\mathcal{O}$ is \emph{robust} if it is closed under the algebraic operations of $\MIN,\MAX,\SUM$ and numerical complement.
Formally, a class of objectives $\mathcal{O}$ is robust, if for every two objectives $\Obj_1,\Obj_2\in\mathcal{O}$ the four quantitative objectives $\Obj_{\MIN},\Obj_{\MAX},\Obj_{\SUM}$ and $\Obj^-$ are in $\mathcal{O}$ (such that for every $\ell\in(\R^k)^\omega$ and $\OP\in\{\MIN,\MAX,\SUM\}$:
$\Obj_{\OP}(\ell) = \OP(\Obj_1(\ell),\Obj_2(\ell))$ and $\Obj^-(\ell) = -\Obj_1(\ell)$).
We note that in~\cite{Weighted-Automaton,mean-payoff-Automaton-Expressions}, Chatterjee et al. gave a broader definition for robustness of quantitative objectives,
but since the concrete objectives that we consider in this paper are robust according to both definitions, we prefer to use the narrower (and simpler) notion of robustness.

\Heading{Games on graph.}
A game graph is a directed graph $G=(V=V_1\cup V_2,v_0, E, w:E\to\Q^k)$, where
$V$ is the set of vertices;
$V_i$ is the set of \emph{player i} vertices;
$v_0$ is the \emph{initial vertex};
$E\subseteq V\times V$ is the set of edges; and
$w:E\to \Q^k$ is a multidimensional weight function (e.g., see Figure~\ref{fig:GameGraph}).
\begin{figure}
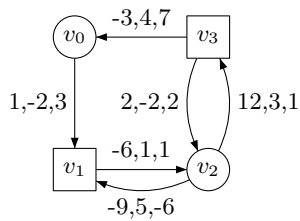

\begin{center}
    \unitlength=2pt
    \begin{gpicture}(0, 35)
    \thinlines

    \node[Nw=8.0,Nh=8.0](A1)(0,0){$v_2$}
    \node[Nw=8.0,Nh=8.0,Nmr=0](A2)(-25,0){$v_1$}
    \node[Nw=8.0,Nh=8.0,Nmr=0](A3)(0,25){$v_3$}
    \node[Nw=8.0,Nh=8.0](A4)(-25,25){$v_0$}

    \drawedge[curvedepth=-4,ELside=r](A1,A3){12,3,1}
    \drawedge[curvedepth=-4,ELside=r](A3,A1){2,-2,2}
    \drawedge[ELside=r](A3,A4){-3,4,7}
    \drawedge[ELside=r](A4,A2){1,-2,3}
    \drawedge(A2,A1){-6,1,1}
    \drawedge[curvedepth=4,ELside=l](A1,A2){-9,5,-6}
    \end{gpicture}
\end{center}
	\caption{Game graph $G$. Player 1 owns the round vertices.}
	\label{fig:GameGraph}
\end{figure}
A \emph{play} is an infinite sequence of rounds.
In the first round a pebble is placed on the initial vertex and in every round the player who owns the vertex of the pebble advances the pebble to an adjacent vertex.
Hence, a play corresponds to an infinite path in the graph that begins in $v_0$ and the labeling of the play is the corresponding infinite sequence of weight vectors.
A game graph is a \emph{one-player game} if only one of the players has a vertex with out-degree more than one.

\Heading{Strategies.}
A \emph{strategy} is a recipe for determining the next move based on the \emph{history} of the play.
A \emph{player-$i$ strategy} is a function $\sigma : V^* V_i \to V$, such that for every finite path $\pi$ that ends in vertex $v$ we have $(v,\sigma(\pi))\in E$.
A strategy has \emph{finite memory} if it can be implemented by a Moore machine $(M,m_0,\alpha_n,\alpha_u)$,
where $M$ is a finite set of memory states, $m_0$ is the initial memory state,
$\alpha_u : M \times V \to M$ is the update function, and $\alpha_n : M \times V_i \to V$ is the next vertex function.
If a play prefix is in state $v_i$ and memory state $M$, then the strategy choice for the next vertex is $v = \alpha_n(M,v_i)$ and the memory is updated to $\alpha_u(M,v_i)$.
A strategy is \emph{memoryless} if it depends only in the current location of the pebble. Formally a player-$i$ memoryless strategy is a function $\sigma: V_i \to V$.
(We note that a memoryless strategy is also a finite-memory strategy.)

We denote the set of all player-$i$ strategies by $\Strategies_i$ and we
denote the set of all player-$i$ finite memory strategies by $\FM_i$.

\Heading{Game graph according to a finite-memory strategy.}
For a game graph $G=(V=V_1\cup V_2,E,w)$ and a player-1 finite-memory strategy $\sigma = (M,m_0,\alpha_u,\alpha_n)$, we denote the \emph{game graph according to strategy} $\sigma$ by $G^\sigma$, and we define it as follows:
\begin{itemize}
\item The vertices of $G^\sigma$ are the Cartesian product $V\times M$; player-$i$ vertices are $V_i \times M$; and
the initial vertex of $G^\sigma$ is $(v_0,m_0)$.
\item For a player-1 vertex $(v,m)$, the only successor vertex is $(\alpha_n(v,m),\alpha_u(v,m))$.
For a player-2 vertex $(v,m)$ the set of successor vertices is
$\{(u,n) \mid (v,u)\in E\mbox{ and } \alpha_u(v,m)=n\}$.
\end{itemize}
We note that the out-degree of all player-1 vertices is one, and thus $G^\sigma$ is a \emph{one-player game graph}.
The main property of graphs according to a finite-memory strategy is that every infinite path in $G^\sigma$ corresponds to a play that is consistent with $\sigma$ in $G$.
A game graph according to a memoryless strategy is a special case of games according to finite-memory strategies.
In this case, the game graph is obtained from $G$ by removing all the player's out-edges that are not chosen by the memoryless strategy.
\begin{examp}
Consider the game graph from Figure~\ref{fig:GameGraph} and consider a player-1 strategy $\sigma$ that in vertex $v_2$ moves the pebble to $v_3$ if $v_2$ was visited an odd number of times and otherwise it moves the pebble to $v_1$.
For example, in the first time that $v_2$ is visited, player 1 moves the pebble to $v_3$, in the second time he will move the pebble to $v_1$, in the third time to $v_3$ and so on.
The strategy $\sigma$ requires one bit of memory (i.e., $M = \{0,1\}$), and $G^\sigma$ is illustrated in Figure~\ref{Fig:GameAccordingTosigma} (the labeling of the nodes represents the memory state).
In $G^\sigma$ all the choices are done by player 2.
\begin{figure}
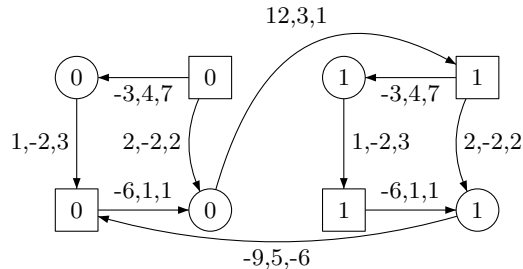

    \unitlength=2pt
\begin{center}
    \begin{gpicture}(0, 35)
    \thinlines

    \node[Nw=8.0,Nh=8.0](A1)(0,0){$0$}
    \node[Nw=8.0,Nh=8.0,Nmr=0](A2)(-25,0){$0$}
    \node[Nw=8.0,Nh=8.0,Nmr=0](A3)(0,25){$0$}
    \node[Nw=8.0,Nh=8.0](A4)(-25,25){$0$}

    \node[Nw=8.0,Nh=8.0](mA1)(50,0){$1$}
    \node[Nw=8.0,Nh=8.0,Nmr=0](mA2)(25,0){$1$}
    \node[Nw=8.0,Nh=8.0,Nmr=0](mA3)(50,25){$1$}
    \node[Nw=8.0,Nh=8.0](mA4)(25,25){$1$}


\drawqbedge(A1,12.5,50,mA3){12,3,1}

    \drawedge[curvedepth=-4,ELside=r](A3,A1){2,-2,2}
    \drawedge(A3,A4){-3,4,7}
    \drawedge[ELside=r](A4,A2){1,-2,3}
    \drawedge(A2,A1){-6,1,1}

    \drawedge[curvedepth=-4,ELside=l](mA3,mA1){2,-2,2}
    \drawedge(mA3,mA4){-3,4,7}
    \drawedge(mA4,mA2){1,-2,3}
    \drawedge(mA2,mA1){-6,1,1}
    \drawedge[curvedepth=6,ELside=l](mA1,A2){-9,5,-6}
    \end{gpicture}
\end{center}
	\caption{Game graph $G$ according to strategy $\sigma$.}
	\label{Fig:GameAccordingTosigma}
\end{figure}
\end{examp}

\Heading{Values of strategies and games.}
A tuple $(\sigma,\tau)$ of player-1 and player-2 strategies (respectively) uniquely defines a play $\pi_{\sigma,\tau}$ in a given graph.
For a game graph $G$, a quantitative objective $\Obj$ and a tuple of strategies $(\sigma,\tau)$ we denote $\Val_{\sigma,\tau} = \Obj(\pi_{\sigma,\tau})$.
In this paper, we assume that player 1 wishes to minimize\footnote{Since we consider robust objectives, then the same results hold when player-1 goal is to maximize the value of the objective.} the value of the quantitative objective, and we define the 
\emph{value} of a player-1 strategy $\sigma$ to be $\ValP_{\sigma} = \sup_{\tau\in \StrategiesTwo} \Val_{\sigma,\tau}$. (Intuitively, this is the maximal value that player 2 can achieve against strategy $\sigma$.)
The \emph{minimal value of a game} is defined as $\inf_{\sigma\in\FiniteOne} \ValP_\sigma$.
Intuitively, the minimal value of a game is the minimal value that player 1 can ensure by a finite-memory strategy.

\Heading{Quantitative and boolean analysis}
For a given game graph, objective, and a rational threshold $r\in\Q$:
The \emph{quantitative analysis} task is to compute the minimal value of the game that can be enforced by a finite-memory strategy.
The \emph{boolean analysis} task is to decide whether there is a player-1 finite-memory strategy $\sigma$ for which $\ValP_{\sigma} \leq r$.
That is, whether player 1 can assure a value of at most $r$ for the objective.

\Heading{Boolean games and winning strategies.}
A \emph{boolean game} is a game on graph equipped with a winning condition $W\subseteq E^\omega$ (that is, a winning condition is a set of infinite paths).
A play $\pi$ is \emph{winning} for player 1 if $\pi\in W$, and a strategy $\sigma$ is a player-1 \emph{winning strategy} if for every player-2 strategy $\tau$ we have $\pi_{\sigma,\tau}\in W$.
For a quantitative objective $\Obj$ and a threshold $\nu\in\R$ we denote by $(\Obj,\nu)$ the boolean winning condition $\{\pi\in E^\omega\mid \Obj(\pi) \leq \nu\}$.

\subsection{One-player game solution}\label{subsec:ConvObj}
In this paper, we consider objectives that have special properties for their one-player game solution and we present a general scheme that lifts a one-player game solution into a two-player game solution.
To formally define the special properties of the solutions, we give the next definitions.

\Heading{Definitions and notions for weighted graphs.}
Let $G=(V,E,w:E\to\Q^k)$ be a $k$-dimensional weighted graph.
The \emph{weight vector} of a finite path $\pi = e_1\dots e_n$ is $w(\pi) =\sum_{i=1}^n w(e_i)$ and the \emph{average weight} of a path is $\Avg(\pi)=\frac{w(\pi)}{|\pi|}$.
For a set of finite paths $\Pi =\{\pi_1,\dots,\pi_n\}$ we denote $\Avg(\Pi)=\{\Avg(\pi_1),\dots,\Avg(\pi_n)\}$.
We denote the set of simple cycles in $G$ by $C(G)$, and we abbreviate $\Avg(G)=\Avg(C(G))$.
For a finite set of vectors $V = \{v_1,v_2,\dots,v_n\}\in \R^k$, we denote $\CONV(V) = \{\sum_{i=1}^n \alpha_i v_i \mid \sum_{i=1}^n \alpha_i = 1 \mbox{ and } \alpha_1,\dots,\alpha_n \geq 0\}$ (see Figure~\ref{fig:Convex}).
\begin{figure}
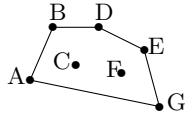

\setlength\unitlength{0.05cm}
\begin{center} 
\begin{gpicture}(0, 35)


\drawpolygon(6,6)(12,20)(24,20)(36,14)(40,-1)

\put(6,6){\circle*{2}}
\put(0,5){\small{A}}

\put(12,20){\circle*{2}}
\put(11,22){\small{B}}

\put(18,10){\circle*{2}}
\put(12,9){\small{C}}

\put(24,20){\circle*{2}}
\put(23,22){\small{D}}

\put(36,14){\circle*{2}}
\put(37,13){\small{E}}

\put(30,8){\circle*{2}}
\put(26,7){\small{F}}

\put(40,-1){\circle*{2}}
\put(42,-2){\small{G}}


\end{gpicture}
\end{center} 
  \caption{$\CONV(A,B,C,D,E,F,G)$ is the polygon $ABDEG$.}
\label{fig:Convex}
\end{figure}
We abbreviate $\CONV(G) = \CONV(\Avg(G))$.
An $m$-dimensional \emph{simplex} is the set $\Simplex(m) = \{(x_1,\dots,x_m)\in\R^m \mid
x_i \geq 0 \wedge \sum_{i=1}^m x_i = 1\}$.
The \emph{simplex interior} is $\IntSimplex(m) = \{(x_1,\dots,x_m)\in\R^m \mid
x_i > 0 \wedge \sum_{i=1}^m x_i = 1\}$, 
and the \emph{rational interior of a simplex} is
$\RatSimplex(m) = \IntSimplex(m)\cap \Q^m$.
When $m$ is clear from the context we abbreviate $\Simplex(m), \IntSimplex(m)$ and $\RatSimplex(m)$ with $\Simplex, \IntSimplex$ and $\RatSimplex$ (respectively).

\Heading{Solution for one-player game with special properties.}
A solution for a one-player quantitative game is a function $f$ that assigns to every one-player game graph $G$ the maximal value that the player can achieve in graph $G$.
We note that for prefix-independent objectives, a function $f'$ that assigns to each \emph{strongly connected} graph its maximal value uniquely defines the solution function $f$ (since the value of $f$ is the maximal value of $f'$ over all the strongly connected component of the graph).
In this paper, we will consider only prefix-independent objective, hence, we define the special properties of a solution for strongly connected graphs.
The special properties that we consider are:
\begin{enumerate}
\item \emph{First-order definable.}
For every $n\in\Nat$ there is a first-order formula $\zeta_n(\VEC{x_1},\dots,\VEC{x_n},y)$ over $\langle \R,=,<,+,\times\rangle$ such that for every graph $G$ with $\Avg(G) = \{\VEC{x_1},\dots,\VEC{x_n}\}$ we have
$f(G) = y$ if and only if $\zeta_n(\VEC{x_1},\dots,\VEC{x_n},y)$ holds.
In addition we require $\zeta_n$ to be computable from $n$.
In the sequel, we write $y=\zeta_n(\VEC{x_1},\dots,\VEC{x_n})$ instead of $\zeta_n(\VEC{x_1},\dots,\VEC{x_n},y)$.
\item \emph{Monotone in $\CONV(G)$.} If for two (strongly connected) graphs $H$ and $G$ we have $\CONV(G) \subseteq \CONV(H)$, then $f(G) \leq f(H)$.
As a consequence, we get that for a $k$-dimensional objective, $f$ is a function from $(\R^k)^*$ to $\R$, and that $f(G) \equiv g(\CONV(G))$ for some function $g:(\R^k)^*\to \R$. Hence, by abusing the notation, we sometime write $f(\CONV(G))$ instead of $f(G)$. 
\item \emph{Continuous function.} $f$ is a continuous function.
Formally, if $f$ is the solution for a $k$-dimensional objective, then for every $n\in\Nat$ the function $\zeta_n : (\R^k)^n \to \R$ is a continuous function, i.e., for every $\epsilon > 0$ there exists $\delta > 0$ such that for every two vectors $A,B\in(\R^k)^n$ with $|A - B| < \delta$ it holds that $|\zeta_n(A)-\zeta_n(B)| < \epsilon$.
\end{enumerate}
We will show computability for the quantitative analysis problem for objectives that have a solution that satisfies the above three properties.
We also consider a fourth special property, and we will show decidability for the boolean analysis problem for objectives that have a solution that satisfies all four properties.
\begin{enumerate}
\setcounter{enumi}{3}
\item \emph{Fourth property.}
A solution $f=\{\zeta_1,\dots,\zeta_n,\zeta_{n+1},\dots\}$ satisfies the fourth property if the next problem is decidable (for the set $\{\zeta_1,\dots,\zeta_n,\zeta_{n+1},\dots\}$):
\begin{itemize}
\item Input: a threshold $\nu\in\Q$ and a set of $n$ matrices $A_1,\dots,A_n$, where $A_i$ is a $k\times m_i$ matrix for some $m_i \in\Nat$.
\item Task: determine if the inequality $\zeta_n(A_1 \cdot x_1,\dots,A_n \cdot x_m)\leq \nu$ subject to $x_i \in \RatSimplex(m_i)$ is feasible (note that the result of the multiplication $A_i\cdot x_i$ is a vector of size $k$).
\end{itemize}
\end{enumerate}
In the next example we demonstrate the above properties.
\begin{examp}\label{examp:Props}
Consider the two-dimensional one-player solution function
$f(G) = \max_{(x,y)\in \CONV(G)} [\max(x+y+10,-x+y+10,\min(-x+y-10,x+y-10))]$.
We demonstrate that $f$ is first-order definable by giving the explicit formula for $\zeta_2$, that is, the formula for a (strongly connected) graph with only two simple cycles with average weights $(x_1,y_1)$ and $(x_2,y_2)$.
\begin{quote}
$\zeta_2(x_1,y_1,x_2,y_2,r) \equiv$\\
$\forall \alpha_1,\alpha_2,x,y (\alpha_1 \geq 0 \wedge \alpha_2 \geq 0 \wedge (\alpha_1 + \alpha_2 = 1)\wedge (x = \alpha_1 x_1 + \alpha_2 x_2 ) \wedge (y = \alpha_1 y_1 + \alpha_2 y_2))\shortrightarrow$\\
$r\geq \max(x+y+10,-x+y+10,\min(-x+y-10,x+y-10))$\\
$\wedge$\\
$\exists \alpha_1,\alpha_2,x,y (\alpha_1 \geq 0 \wedge \alpha_2 \geq 0 \wedge (\alpha_1 + \alpha_2 = 1)\wedge (x = \alpha_1 x_1 + \alpha_2 x_2 ) \wedge (y = \alpha_1 y_1 + \alpha_2 y_2))\wedge$\\
$r = \max(x+y+10,-x+y+10,\min(-x+y-10,x+y-10))$
\end{quote}
(Technically $\max$ and $\min$ are not in $\langle \R,<,+,\times \rangle$, but they are trivially definable in this vocabulary.)
Clearly if for two graphs we have $\CONV(G_1) \subseteq \CONV(G_2)$, then $f(G_1)\leq f(G_2)$ (hence, $f$ is monotone), and $\zeta_2$ is obviously a continuous function (and in general $\zeta_n$ is also continuous).
Hence, $f$ satisfies Properties~1-3.
In Figure~\ref{fig:ExampleOfProp} we illustrate the geometrical interpretation of Property~2, namely, the fact that the value of $f$ depends only in $\CONV(G)$.
The equality $\max(x+y+10,-x+y+10,\min(-x+y-10,x+y-10)) = 0$ is represented by the thick line.
The points that are connected by the dotted line represent the weights of the simple cycles of a strongly connected graph $G_1$ and the points that are connected by the dashed line represent the weights of the simple cycles of a strongly connected graph $G_2$.
The reader can see that $\CONV(G_1)$ is below the thick line and $\CONV(G_2)$ intersects with it.
Hence, $f(G_1) < 0$ and $f(G_2) > 0$.

\begin{figure}
\setlength\unitlength{0.1cm}
\begin{center} 
\begin{gpicture}(0,35)(0,-10)


\psline[linewidth=0.1pt,linestyle=dashed]{o}(-0.5,-0.8)(0.25,0.5)
\psline[linewidth=0.1pt,linestyle=dashed]{o}(0.25,0.5)(1,-0.5)
\psline[linewidth=0.1pt,linestyle=dashed]{o}(1,-0.5)(1.5,0.3)
\psline[linewidth=0.1pt,linestyle=dashed]{o}(1.5,0.3)(2,0.5)
\psline[linewidth=0.1pt,linestyle=dashed]{o}(2,0.5)(2,-0.8)
\psline[linewidth=0.1pt,linestyle=dashed]{o}(2,-0.8)(-0.5,-0.8)


\psline[linewidth=0.4pt,linestyle=dotted]{*}(-2.75,-1.0)(-0.3,0.2)
\psline[linewidth=0.4pt,linestyle=dotted]{*}(-0.3,0.2)(0.3,0.2)
\psline[linewidth=0.4pt,linestyle=dotted]{*}(0.3,0.2)(2.75,-1)
\psline[linewidth=0.4pt,linestyle=dotted]{*}(2.75,-1)(-2.75,-1.0)

\psline[linewidth=0.2pt]{<->}(0,-1.2)(0,1.7)
\psline[linewidth=0.2pt]{<->}(-2.5,0)(2.5,0)

\psline[linewidth=2pt,linestyle=dotted](-2,1.0)(-2.4,1.4)
\psline[linewidth=2pt](-2,1.0)(-1,0)

\psline[linewidth=2pt](-1,0)(0,1)

\psline[linewidth=2pt](0,1)(1,0)

\psline[linewidth=2pt](1,0)(2,1.0)
\psline[linewidth=2pt,linestyle=dotted](2,1.0)(2.4,1.4)

\end{gpicture}
\end{center} 
  \caption{}
\label{fig:ExampleOfProp}
\end{figure}
\end{examp}
\subsection{Informal overview of the solution for two-player games}\label{sect:Overview}
The key notion for our solution is games according to strategies.
When a one-player solution $f$ is given, the boolean analysis problem amounts to determining whether there is a finite-memory strategy $\sigma$ such that for every strongly-connected component (SCC) $S \in G^\sigma$ it holds that $f(\CONV(S)) \leq \nu$.
In Lemma~\ref{lem:WinningRegion} we show that w.l.o.g we may assume that for any 
$\sigma$ the graph $G^\sigma$ is strongly connected.
Hence, in Section~\ref{sect:SPANGame} we investigate the set $\{\CONV(G^\sigma)\mid \sigma\in\FiniteOne\}$ and obtain a computable characterization for it.
In Section~\ref{sect:Generic} we exploit the properties of the one-player solution and the results of Section~\ref{sect:SPANGame} and we obtain a first-order formula over \textbf{rationals} that computes the values that player 1 can enforce.
We use the fact that $f$ is continuous to show that the formula has the same infimum over rationals and reals, and hence, due to Tarski's Theorem the infimum value is computable.
We also show that if Property~4 holds, then one can effectively determine whether the formula has an assignment that gives a value of at most $\nu$.
In Section~\ref{sect:OnePlayerMPExpProp} we apply these results for mean-payoff expressions.
We show that their one-player solution satisfies Properties~1-3, and that it satisfies property~4 if and only if $\HQ$ is decidable.

\section{$\CONV$ Cycles Problem}\label{sect:SPANGame}
In this section, we consider the next problem:
\begin{prob}[$\CONV$ cycles problem]\label{prob:ConvCycles}
\begin{itemize}
\item Input: a $k$-dimensional game graph $G$ and a set of $k$-dimensional vectors $\VectorSetV$.
\item Task: determine whether there is a player-1 finite-memory strategy $\sigma$ such that $\CONV(G^\sigma) \subseteq \CONV(\VectorSetV)$.
(We call such strategy a \emph{realizing} strategy.)
\end{itemize}
\end{prob}
We first present the solution for the above problem, and then we show how to find all the sets of vectors for which there is a realizing player-1 finite-memory strategy.

The solution for Problem~\ref{prob:ConvCycles} relies on the next lemma.
\begin{lem}\label{lem:ConvCyclesMemoryless}
For a game graph $G$ and a set of vectors $\VectorSetV$, there exists a player-1 finite-memory strategy $\sigma$ for which $\CONV(G^\sigma) \subseteq \CONV(\VectorSetV)$ iff for every player-2 memoryless strategy $\tau_i$ there exists a player-1 finite-memory strategy $\sigma_i$ such that $\CONV((G^{\tau_i})^{\sigma_i}) \subseteq \CONV(\VectorSetV)$.
\end{lem}
\begin{proof}
The proof for the direction from left to right is trivial (since
any cycle in $(G^{\tau_i})^\sigma$ is also a cycle in $G^\sigma$).

Our proof for the converse direction is inspired by~\cite{HalfPositional}, and the key intuition of the proof is the following.
Let $v$ be a player-2 vertex, with two out-edges $e_1$ and $e_2$, and let $G_1 = G - \{e_1\}$ and $G_2 = G - \{e_2\}$.
Suppose that player 1 has two finite-memory strategies $\sigma_1$ and $\sigma_2$ such that $\CONV(G_i ^ {\sigma_i} ) \subseteq \CONV(\VectorSetV)$ (for $i=1,2$).
Then player 1 can combine the two strategies over $G-\{e_1\}$ and $G-\{e_2\}$, and he can obtain a finite-memory strategy $\sigma$, such that each simple cycles in $G^\sigma$ is a convex combination of cycles from $G_1^{\sigma_1}$ and $G_2^{\sigma_2}$, and hence, $\CONV(G^\sigma) \subseteq\CONV(\VectorSetV)$.
Hence, either player 1 has a realizable strategy for $G$ or he does not have realizable strategy for $G_1 = G - \{e_1\}$ or for $G_2 = G - \{e_2\}$.
Since this holds for every player-2 state, the proof follows.

In order to formally prove the key intuition we claim that if player 1 has a realizable strategy $\sigma_i$ against any player-2 memoryless strategy, then he has a realizable strategy $\sigma$ that satisfies $\CONV(G^\sigma)\subseteq \CONV(\VectorSetV)$, and we prove the claim by induction on the number of player-2 vertices with out-degree greater then one.
The base case, where all of player-2 vertices have out-degree one, is trivial.
For the inductive step, let us assume that there is a player-2 vertex $v$ with out-edges $e_1$ and $e_2$ (if there is no such vertex, then we are in the base case).
For $i=1,2$, let $G_i$ be $G - \{e_i\}$.
If player-2 has a violating memoryless strategy in either $G_1$ or $G_2$, then surely this is also a violating memoryless strategy for $G$, and the claim follows.
Otherwise, we construct a realizable player-1 strategy in $G$ in the following way.
For $i=1,2$, let $\sigma_i$ be a finite-memory player-1 realizable strategy in $G_i$,
If in $\sigma_1$ (resp., $\sigma_2$), the vertex $v$ is unreachable then it is surely a winning strategy also for $G$.
Otherwise, there exists a memory state $m$ such that $(m,v)$ is a vertex in $G_1^{\sigma_1}$, and we denote by $\sigma_1 '$ the strategy that is obtained by changing $\sigma_1$ initial memory state to $m$.
We construct $\sigma$ in the following way.
The memory structure of $\sigma$ is a tuple $(M_1,M_2,\{1,2\})$ where $M_1$ is the memory structure of $\sigma_1 '$, $M_2$ is the memory structure of $\sigma_2$, and the third value in the tuple indicates if we are playing according to $\sigma_1 '$ or $\sigma_2$.
At the beginning of a play, $\sigma$ decides according to $\sigma_2$ (and updates $M_2$ accordingly).
If $\sigma$ decides according to $\sigma_2$ and edge $e_1$ is visited, then $\sigma$ decides according to $\sigma_1 '$ (and updates $M_1$ accordingly), until edge $e_2$ is visited, and then $\sigma$ again decides according to $\sigma_2$, and so on.
We note that $\sigma$ is a finite-memory strategy,
and that any simple cycle in $G^\sigma$ is a composition of simple cycles from $G_1^{\sigma_1 '}$ and $G_2^{\sigma_2}$.
Hence, the average weight of any simple cycle in $G^\sigma$ is $\lambda x_1 + (1-\lambda) x_2$ for some $\lambda \in [0,1]$ and $x_i \in \Avg(G_i) \subseteq\CONV(\VectorSetV)$.
And thus, a convex combination of $x_1$ and $x_2$ is also in $\CONV(\VectorSetV)$, and we get that $\CONV(G^\sigma)\subseteq \CONV(\VectorSetV)$.
Therefore, $\sigma$ is a realizing strategy and the proof is completed.
\pfbox
\end{proof}
We now wish to characterize all the sets of vectors that have a realizing strategy. For this purpose we give the next definition.
For a player-2 memoryless strategy $\tau$, let $\Pi_e^\tau$ be the (finite) set of Eulerian cyclic paths in $G^\tau$, that is $\Pi_e^\tau$ contains only cyclic paths that visit every edge at most once.
For every path $\pi \in \Pi_e^\tau$, let $c_1,\dots,c_t$ be the simple cycles that occur in $\pi$ and we associate a $t\times k$ matrix $A_\pi$ to every path $\pi$ such that the $i$-th column of the matrix is $\Avg(c_i)$.
We observe that
\[\{\Avg(\pi)\mid \mbox{ $\pi$ is a cyclic path in $G^\tau$}\}=
\bigcup_{\pi\in\Pi_e^\tau} \{A_\pi \cdot x \mid x\in\RatSimplex\}\]
The next lemma shows how to compute the realizable sets of vectors.
\begin{lem}\label{lem:RealizableSetComputable}
Let $G$ be a $k$-dimensional graph, and let $\tau_1,\dots,\tau_\ell$ be the (finitely many) player-2 memoryless strategies in $G$.
A set of vectors $\VectorSetV\subseteq \R^k$ is realizable if and only if there exist $x_1,\dots,x_\ell \in \R^k$ such that $x_i \in \bigcup_{\pi\in\Pi_e^{\tau_i}} \{A_\pi \cdot x \mid x\in\RatSimplex\}$ (for every $i\in\{1,\dots,\ell\}$) and $\CONV(x_1,\dots,x_\ell) \subseteq \CONV(\VectorSetV)$.
\end{lem}
\begin{proof}
First we characterize the realizable vector sets when a player-2 memoryless strategy $\tau$ is given, that is, we characterize the realizable vectors in a one-player game.
A finite-memory strategy $\sigma$ in a one-player graph $G^\tau$ is an ultimately periodic infinite path, and $(G^\tau)^\sigma$ is a \emph{lasso shaped graph} with exactly one cycle.
The cycle of $(G^\tau)^\sigma$ is obviously a cyclic path in $G^\tau$, and thus $\VectorSetV$ is realizable in $G^\tau$ iff there is a cyclic path $\pi$ in $G^\tau$ with $\Avg(\pi) \in  \CONV(\VectorSetV)$.

Hence, by Lemma~\ref{lem:ConvCyclesMemoryless}, we get that $\VectorSetV$ is realizable iff for every player-2 memoryless strategy $\tau_i$ there is a cyclic path $\pi_i$ in $G^{\tau_i}$ with $\Avg(\pi_i)\in  \CONV(\VectorSetV)$.
Since such witness $\pi_i$ exists iff there exists $x_i \in \bigcup_{\pi\in\Pi_e^{\tau_i}} \{A_\pi \cdot x \mid x\in\RatSimplex\}$ with $\Avg(\pi_i ) =x_i$, then the proof is completed.
\pfbox
\end{proof}
In the next example we illustrate the geometrical interpretation of Lemma~\ref{lem:RealizableSetComputable}.
\begin{examp}\label{examp:GeoInturp}
Consider the game graph $G$ in Figure~\ref{fig:GeoGraph}, where the box vertices are controlled by player 2.
Player 2 has two possible memoryless strategies, namely, $\tau_1$ that follows the edge $v_0 \to v_1$ and $\tau_2$ that follows $v_0\to v_4$.
In $G^{\tau_1}$ the set of Eulerian cyclic paths $\Pi_e^{\tau_1}$
contains all cyclic sub-paths of the Eulerian cyclic path $v_1\to v_2 \to v_2 \to v_3\to v_3 \to v_1$.
Hence, the average weight of any infinite lasso path in $G^{\tau_1}$ is a convex combination of $\Avg(v_1\to v_2 \to v_3\to v_1)$, $\Avg(v_2\to v_2)$ and $\Avg(v_3\to v_3)$ (points $D$, $F$ and $E$ in Figure~\ref{fig:GeoAxis}).
In $G^{\tau_2}$, an Eulerian cyclic path is either a sub-path of
$v_4 \to v_5 \to v_5 \to v_4$ or the path $v_6 \to v_6$.
Hence, the average weight of any infinite lasso path is either a convex combination of $\Avg(v_4\to v_5\to v_4)$ and $\Avg(v_5\to v_5)$ (points $A$ and $B$ in Figure~\ref{fig:GeoAxis}), or it is $\Avg(v_6\to v_6)$ (point $C$ in Figure~\ref{fig:GeoAxis}).
By Lemma~\ref{lem:RealizableSetComputable}, we get that a set of vectors $\VectorSetV$ is realizable if and only if $\CONV(\VectorSetV)$ intersects with the polygon $DEF$ and with either the line $AB$ or with the point $C$ (or with both).
\end{examp}
\begin{figure}
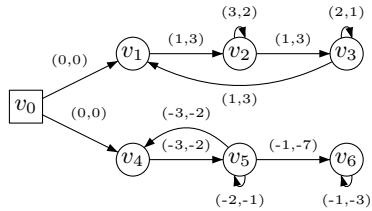

\begin{center}
    \unitlength=4pt
    \begin{gpicture}(30,15)(0,-10)
    \thinlines

    \node[Nw=3.1,Nh=3.1,Nmr=0](A1)(0,0){$v_0$}

    \node[Nw=3.1,Nh=3.1](B1)(10,5){$v_1$}
    \node[Nw=3.1,Nh=3.1](B2)(20,5){$v_2$}
    \node[Nw=3.1,Nh=3.1](B3)(30,5){$v_3$}

    \node[Nw=3.1,Nh=3.1](C1)(10,-5){$v_4$}
    \node[Nw=3.1,Nh=3.1](C2)(20,-5){$v_5$}
    \node[Nw=3.1,Nh=3.1](C3)(30,-5){$v_6$}

    \drawedge(A1,B1){\tiny{(0,0)}}
    \drawedge(A1,C1){\tiny{(0,0)}}

    \drawedge(B1,B2){\tiny{(1,3)}}
    \drawedge(B2,B3){\tiny{(1,3)}}
    \drawedge[curvedepth=3,ELside=l](B3,B1){\tiny{(1,3)}}
    \drawloop[loopdiam = 1](B2){\tiny{(3,2)}}
    \drawloop[loopdiam = 1](B3){\tiny{(2,1)}}

    \drawedge(C1,C2){\tiny{(-3,-2)}}
    \drawedge(C2,C3){\tiny{(-1,-7)}}
    \drawedge[curvedepth=-3,ELside=r](C2,C1){\tiny{(-3,-2)}}
    \drawloop[loopangle=270,loopdiam = 1](C2){\tiny{(-2,-1)}}
    \drawloop[loopangle=270,loopdiam = 1](C3){\tiny{(-1,-3)}}


    \end{gpicture}
\end{center}
	\caption{Game graph $G$.}
	\label{fig:GeoGraph}
\end{figure}
\begin{figure}
\setlength\unitlength{0.1cm}
\begin{center} 
\begin{gpicture}(0,35)(0,-10)

\psline[linewidth=0.5pt]{<->}(0,-1.3)(0,1.3)
\psline[linewidth=0.5pt]{<->}(-1.3,0)(1.3,0)


\psline[linewidth=0.5pt,linestyle=dashed](0.5,1.5)(1.5,1)
\psline[linewidth=0.5pt,linestyle=dashed](1.5,1)(1,0.5)
\psline[linewidth=0.5pt,linestyle=dashed](1,0.5)(0.5,1.5)

\psline[linewidth=0.5pt,linestyle=dashed](-1,-0.5)(-1.5,-1)

\put(-15,-10){\circle*{1}}
\put(-16.3,-12.5){\tiny{$A$}}

\put(-10,-5){\circle*{1}}
\put(-11.3,-7.5){\tiny{$B$}}

\put(-5,-5){\circle*{1}}
\put(-6.3,-7.5){\tiny{$C$}}

\put(5,15){\circle*{1}}
\put(3.7,12.5){\tiny{$D$}}

\put(10,5){\circle*{1}}
\put(8.7,2.5){\tiny{$E$}}

\put(15,10){\circle*{1}}
\put(13.7,7.5){\tiny{$F$}}


\end{gpicture}
\end{center} 
\caption{}
\label{fig:GeoAxis}
\end{figure}

\section{Generic Solution for Games with Quantitative Objectives}\label{sect:Generic}
In this section we solve the quantitative analysis problem for games with quantitative objectives that satisfy Properties~1-3 and we solve the boolean analysis problem for objectives that satisfy Properties~1-4.
We first give a conceptual (i.e., not always computable) solution for the boolean analysis problem, and then extend the solution for the quantitative analysis problem.

An equivalent formulation for the boolean analysis problem is to ask whether for a game graph $G$ and a threshold $\nu$ there is a player-1 (finite-memory) strategy $\sigma$ such that the one-player solution over $G^\sigma$ is at most $\nu$.
By the third property (convex monotonicity), it is enough to determine whether there is $\sigma$ such that for every SCC $S$ of $G^\sigma$ it holds that $f(\CONV(S)) \leq \nu$ (where $f$ is the solution for the one-player game).
However, we first show how to determine whether there is $\sigma$ such that $f(\CONV(G^\sigma)) \leq \nu$ and only then solve the original problem.
\begin{lem}\label{lem:GeneralG}
Let $f$ be a one-player solution that satisfies Properties~1-3.
Then $\inf_{\sigma\in\FiniteOne}f(\CONV(G^\sigma))$ is computable (when the input is a game graph $G$).
If $f$ also satisfies Property~4, then the problem of determining whether there is a player-1 finite-memory strategy $\sigma$ such that $f(\CONV(G^\sigma))\leq \nu$ is decidable (when the input is $G$ and $\nu$).
\end{lem}
\begin{proof}
Let $\tau_1,\dots,\tau_m$ be all player-2 memoryless strategies in $G$ (note that $m$ is at most exponential in $|G|$).
By Lemma~\ref{lem:RealizableSetComputable}, and by the monotonicity of $f$, there is a player-1 strategy $\sigma$ that satisfies $f(\CONV(G^\sigma))\leq \nu$ if and only if there are matrices $A_{\pi_1},\dots,A_{\pi_m}$ and vectors $x_1,\dots,x_m$ such that $\pi_i \in \Pi_e^{\tau_i}$ and $x_i\in\RatSimplex$ 
and $\zeta_m(A_{\pi_1}\cdot x_1, \dots,A_{\pi_m}\cdot x_m) \leq \nu$.
For every $\tau_i$, the set $\Pi_e^{\tau_i}$ is finite (and at most of exponential size).
Hence, we can enumerate all $m$-tuples of $\Pi_e^{\tau_1} \times \dots \times \Pi_e^{\tau_m}$ and check if for at least one tuple there is a solution to the inequality $\zeta_m(A_{\pi_1}\cdot x_1, \dots,A_{\pi_m}\cdot x_m) \leq \nu$.
If $f$ satisfies Property~4, then for a given
$\pi_1,\dots,\pi_m$ we can effectively check if the inequality is satisfiable.
Hence, we can effectively determine whether there is $\sigma$ such that $f(\CONV(G^\sigma))\leq \nu$.
Moreover, for a given $\pi_1,\dots,\pi_m$ the expression 
$\inf_{x_1,\dots,x_m\in \IntSimplex} \zeta_m(A_{\pi_1}\cdot x_1, \dots,A_{\pi_m}\cdot x_m)$
is first-order definable (recall that $\zeta_m$ is first-order definable) over $\langle \R,<,+,\times \rangle$ (note that $x_i$ ranges over $\IntSimplex$ and not over $\RatSimplex$) and therefore, by Taski's Theorem~\cite{Tarski} its value is computable.
Since $\zeta_m$ is continuous we get that
$\inf_{x_1,\dots,x_m\in \IntSimplex} \zeta_m(A_{\pi_1}\cdot x_1, \dots,A_{\pi_m}\cdot x_m)
=
\inf_{x_1,\dots,x_m\in \RatSimplex} \zeta_m(A_{\pi_1}\cdot x_1, \dots,A_{\pi_m}\cdot x_m)$.
Finally, by Lemma~\ref{lem:RealizableSetComputable}, we have that
$\inf_{\sigma\in\FiniteOne} f(\CONV(G^\sigma)) =
\inf_{\pi_1\in \Pi_e^{\tau_1},\dots,\pi_m\in \Pi_e^{\tau_m}}
\inf_{x_1,\dots,x_m\in \RatSimplex} \zeta_m(A_{\pi_1}\cdot x_1,\dots,A_{\pi_m}\cdot x_m)$,
and since $\Pi_e^{\tau_i}$ is finite we get that $\inf_{\sigma\in\FiniteOne} f(\CONV(G^\sigma))$ is computable.
\pfbox
\end{proof}
Before presenting the algorithm for the boolean analysis problem we recall the (standard) definitions of \emph{winning regions} and \emph{attractors}.
Let $G$ be a game graph with an initial vertex $v_0$, and let $v$ be an arbitrary vertex in $G$. We denote by $(G,v)$ the game graph that is formed from $G$ by changing the initial vertex to $v$.
We say that a vertex $v$ is in \emph{player-1 winning region} (denoted by $\WinOne$) if player 1 wins in $(G,v)$ (that is, player 1 has a finite-memory strategy that assures a value at most $\nu$ to the objective).
The \emph{player-1 attractor set} of a vertex $v$ (denoted by $\AttrOne{v}$) contains all the vertices from which player 1 can force reachability to $v$ (after finite number of rounds).
It is well known that the attractor set of a vertex is computable (even in linear time) and that player 1 can force reachability by a finite-memory strategy (in fact, even by a memoryless strategy).
The next remark shows another important property of attractors and winning regions.
\begin{rem}\label{rem:Attractor}
Let $G$ be a game graph over a boolean objective that is formed by a quantitative objective with a solution function $f$ and a threshold $\nu$, and let $v$ be a vertex in $G$.
Then for every vertex $u\notin \AttrOne{v}$, if $\sigma$ is a finite-memory player-1 strategy for $(G,u)$, then $\sigma$ is a winning strategy in $(G-\AttrOne{v},u)$.
\end{rem}
\begin{proof}
We denote $H = G-\AttrOne{v}$ and we observe that $(H,u)^\sigma$ is a subgraph of $(G,u)^\sigma$.
Hence, for every SCC $S\in H$ there is a corresponding SCC $S '\in G$ such that $f(\CONV(S'))\leq \nu$.
Since $\CONV(S) \subseteq \CONV(S')$ and by the monotonicity of $f$ we get that $f(\CONV(S))\leq \nu$ and therefore $\sigma$ is a winning strategy in $(H,u)$.
\pfbox
\end{proof}

Algorithm~\ref{algo:WinningRegion} computes player-1 winning region, and we prove its correctness in Lemma~\ref{lem:WinningRegion}
\begin{algorithm}{{\sc WinningRegion}$(G,f,\nu)$}
\caption{Player-1 winning region computation for quantitative objectives}
\label{algo:WinningRegion}
\begin{algorithmic}
\FOR{$v\in G$}
	\IF{$\exists \sigma$ s.t $f(\CONV((G,v)^\sigma))\leq \nu$}
		\STATE $W \gets \AttrOne{v}$
		\STATE $W \gets W\cup${\sc WinningRegion}$(G - \AttrOne{v} ,f,\nu)$
		\RETURN $W$
	\ENDIF
\ENDFOR
\RETURN $\emptyset$
\end{algorithmic}
\end{algorithm}
\begin{lem}\label{lem:WinningRegion}
Algorithm~\ref{algo:WinningRegion} computes player-1 winning region.
\end{lem}
\begin{proof}
We first prove that in every step of the algorithm, if a vertex $u\in W$, then $u\in\WinOne$.
We prove the assertion by considering the next three cases:
(i)~There is a strategy $\sigma$ for which $f(\CONV((G,u)^\sigma))\leq \nu$.
In this case, for every SCC $S\in (G,u)^\sigma$ we have that
$\CONV(S)\subseteq \CONV((G,u)^\sigma)$ and by the monotonicity of $f$ we get that $f(\CONV(S))\leq \nu$. Hence, $u\in\WinOne$.
(ii)~There is a vertex $v$ and a strategy $\sigma$ s.t $f(\CONV((G,v)^\sigma))\leq \nu$ and $u\in \AttrOne{v}$.
In this case, $v$ is in player-1 winning region and therefore the attractor of $v$ is also in $\WinOne$.
(iii)~For some vertex $v$ we have $u\in ${\sc WinningRegion}$(G - \AttrOne{v} ,f,\nu)$ and $f(\CONV((G,v)^\sigma))\leq \nu$ for some strategy $\sigma$.
By a simple induction on the size of the graph we get that $u$ is in player-1 winning region for the game graph $G - \AttrOne{v}$. The following strategy is a winning strategy for $(G,u)$: (a)~play according to the winning strategy over $G - \AttrOne{v}$; (b)~if the pebble is in vertex $v$, then play according to $\sigma$.
Hence, if $u\in W$, then $u\in\WinOne$ and we get that $\WinOne \supseteq W$.

In order prove the converse direction, we first prove that if $\WinOne \neq \emptyset$, then $W\neq \emptyset$.
Indeed, if $v\in\WinOne$, then for some strategy $\sigma$ we have that for every SCC $S\in (G,v)^\sigma$ it holds that $f(\CONV(S))\leq \nu$.
Let $S '$ be a terminal SCC in $(G,v)^\sigma$ and let $(u,m)$ be a vertex in $S '$ (where $u$ is a vertex in $G$ and $m$ is a memory state of $\sigma$).
Let $\sigma '$ be the strategy that is formed by changing $\sigma$ initial memory state to $m$. Then $(G,u)^{\sigma '} = S '$, and therefore $f(\CONV((G,u)^{\sigma '}))\leq \nu$.
Hence, the \textbf{if} condition in the \textbf{for} loop is satisfied at least once, and $W\neq \emptyset$.
We are now ready to prove that $\WinOne\subseteq W$.
Towards a contradiction we assume the existence of $u\in (\WinOne - W)$.
By the definition of Algorithm~\ref{algo:WinningRegion} it follows that there is a subgraph $H\subseteq G$ such that $u\in H$ and the algorithm returns $\emptyset$ when it runs over $H$.
Hence, player-1 winning region in $H$ is empty (namely, $u\notin \WinOne$ over game graph $H$) and by Remark~\ref{rem:Attractor} we get that $u\notin\WinOne$ in game graph $G$ and the contradiction follows.
Thus $\WinOne \subseteq W$.
\pfbox
\end{proof}
We present a similar algorithm for the computation of quantitative analysis of quantitative objectives.
For this purpose we extend the notion of winning regions to quantitative objectives by defining \emph{value regions}.
For a threshold $\nu$ we say that a vertex $v$ is in \emph{$\nu$ value region} (denoted by $\ValueRegion(\nu)$) if $\inf_{\sigma\in\FiniteOne}\sup_{\tau\in\StrategiesTwo} \Val_{\sigma,\tau} = \nu$ (when the initial vertex of the game is $v$).
Algorithm~\ref{algo:ValueRegion} computes value regions by a call to
{\sc ValueRegion}$(G,f,-\infty)$, and 
its correctness follows by the same arguments as in the proof of Lemma~\ref{lem:WinningRegion}.
\begin{algorithm*}{{\sc ValueRegion}$(G,f,\mathit{ValLowerBound})$}
\caption{Value region computation for quantitative objectives. The algorithm invokes {\sc ValueRegion}$(G,f,-\infty)$.}
\label{algo:ValueRegion}
\begin{algorithmic}
\IF{$G \neq \emptyset$}
\FOR{$v\in G$}
	\STATE $I[v] \gets 
	\max(\inf_{\sigma\in \FiniteOne} f(\CONV((G,v)^\sigma)), \mathit{ValLowerBound})$
\ENDFOR
\STATE $u \gets \operatorname{argmin}_{v\in G} I[v]$
\COMMENT{Choose $u$ s.t $I[u] = \min_{v\in G} I[v]$}
\STATE $\ValueRegion(I[u]) \gets \ValueRegion(I[u]) \cup \AttrOne{u}$
\COMMENT{Add $\AttrOne{u}$ to the value region of $I[u]$}
\RETURN {\sc ValueRegion}$(G - \AttrOne{u},f,I[u])$
\COMMENT{Continue the computation recursively. The new lower bound is $I[u]$.}
\ENDIF
\end{algorithmic}
\end{algorithm*}
We note that if $f$ satisfies Properties 1-3, then by Lemma~\ref{lem:GeneralG}, 
there is an effective procedure to compute $\inf_{\sigma\in \FiniteOne} f(\CONV((G,v)^\sigma))$ (hence, Algorithm~\ref{algo:ValueRegion} can be effectively executed) and if $f$ satisfies Properties 1-4, then by the same lemma we get that there is a procedure to determine whether $f(\CONV((G,v)^\sigma))\leq \nu$ (hence, Algorithm~\ref{algo:WinningRegion} can be effectively executed).
Hence, we get the main result of this section.
\begin{thm}\label{thm:AnalysisOfQuantitativeObj}
Let $f$ be the one-player solution of a quantitative objective.
\begin{itemize}
\item If $f$ satisfies Properties~1-3, then the corresponding quantitative analysis problem is computable.
\item If $f$ satisfies Properties~1-4, then the corresponding boolean analysis problem is decidable.
\end{itemize}
\end{thm}
We note that Theorem~\ref{thm:AnalysisOfQuantitativeObj} provides a recipe for the construction of $\epsilon$-optimal strategies.
If the infimum value that player 1 can achieve is $\nu$, then the process that enumerates all $\sigma\in\FiniteOne$ and halts if the one-player solution of $G^\sigma$ is at most $\nu + \epsilon$ will always terminate.
Similarly, if the boolean analysis problem is decidable, then it is possible to effectively construct a finite-memory strategy that assures the corresponding threshold (we first check if such a strategy exists, and if it does exist, then we enumerate all finite-memory strategies until we find a strategy $\sigma$ such that the solution for $G^\sigma$ is at most $\nu$).
\section{Games with Mean-Payoff Expression Objectives}\label{sect:OnePlayerMPExpProp}
In this section, we give the formal definition of mean-payoff expressions and we use the results of Section~\ref{sect:Generic} to analyze games with mean-payoff expressions.
In Subsection~\ref{subsec:MPExp} we define mean-payoff expressions and show that optimal strategies may require infinite memory.
In Subsection~\ref{subsec:MPExpGames} we analyze mean-payoff expression games.
\subsection{Mean-payoff expression objectives}\label{subsec:MPExp}
The class of mean-payoff expressions is the closure of single dimension mean-payoff objectives under the algebraic operations of $\MIN,\MAX,\SUM$ and numerical complement.
Formally, for an infinite sequence of reals $\rho = a_1,a_2,\dots \in \R^\omega$, we denote $\LimInfAvg(\rho) = \liminf_{n\to\infty} \frac{a_1+\dots+a_n}{n}$ and
$\LimSupAvg(a_1,a_2,\dots) = \limsup_{n\to\infty} \frac{a_1+\dots+a_n}{n}$.
For an infinite sequence of vectors $\rho = v_1,v_2\dots\in(\R^k)^\omega$ we denote be the projection of $\rho$ to the $i$-th dimension by $\rho_i$, and we denote $\LimInfAvg_i(\rho) = \LimInfAvg(\rho_i)$ and $\LimSupAvg_i(\rho) = \LimSupAvg(\rho_i)$.
An \emph{atomic expression} over $\R^k$ is either $\LimInfAvg_i$ or $\LimSupAvg_i$.
If $E_1$ and $E_2$ are expressions, then $-E_1$, $\MAX(E_1,E_2)$, $\MIN(E_1,E_2)$ and $\SUM(E_1,E_2)$ are also expressions.
For a sequence $\rho\in(\R^k)^\omega$ and an expression $E$, the value of $E(\rho)$ is $\LimInfAvg_i(\rho)$ if $E =\LimInfAvg_i$, $\LimSupAvg_i(\rho)$ if $E =\LimSupAvg_i$, $\OP(E_1(\rho),E_2(\rho))$ if $E=\OP(E_1,E_2)$ (for $\OP\in\{\MIN,\MAX,\SUM\})$ and $-E_1(\rho)$ if $E = -E_1$.
Over $\R^2$, a possible expression is $E=\MIN(\LimInfAvg_1,\LimSupAvg_1 + \LimInfAvg_2) + \MAX(\LimInfAvg_1,\LimSupAvg_2)$, and the value of $E$ for the sequence $(-1,1)^\omega$ is $\MIN(-1,-1+1)+\MAX(-1,1) = 0$.

We say that an expression $E$ is of \emph{normal form} if (i)~the numerical complement does not occur in $E$; and (ii)~for every dimension $i$, there is at most one occurrence of an atomic expression $A_i\in\{\LimInfAvg_i,\LimSupAvg_i\}$; and (iii) $E=\MAX(E_1,\dots,E_\ell)$, where $E_i$ is a \emph{max-free} expression (that is, the $\MAX$ operator does not occur in $E_i$).
The next simple lemma shows that w.l.o.g we may consider only games over normal form expressions.
\begin{lem}\label{lem:MPExpNF}
For every $k$-dimensional weighted graph $G$ with a weight function $w$ and an expression $E$, we can effectively construct an $m$-dimensional weight function $w'$ and a normal form expression $F$ such that every infinite path in $G$ gets the same value according to $(E,w)$ and according to $(F,w')$.
\end{lem}
\begin{proof}
We can easily overcome the restriction on the number of atomic expressions per dimension by creating several \emph{copies} of the same dimension (that is, additional dimensions with weights that are identical to the original dimension).
We can create an equivalent numerical complement free expression by the following recursive process.
If $E = -\LimInfAvg_i$ (respectively, $E=-\LimSupAvg_i$), then we multiply all the weights in dimension $i$ by $-1$ and define $F=\LimSupAvg_i$ (resp. $F=\LimInfAvg_i$).
$F$ is equivalent to $E$ since $\LimInfAvg(a_1,a_2,\dots) = -\LimSupAvg(-a_1,-a_2,\dots)$.
If $E=-\OP(E_1,E_2)$, then we recursively change the weights and construct normal form expressions $F_1$ and $F_2$ that are equivalent to $-E_1$ and $-E_2$, and return the normal form expression $F=\OP(F_1,F_2)$.
And we similarly handle the expression $E=\OP(E_1,E_2)$.
Finally, if we have a numerical complement free expression $E$, then we construct an equivalent expression $F=\MAX(F_1,\dots,F_\ell)$, where $F_i$ is a max free expression, by the following recursive procedure:
If $E$ is an atomic expression, then we return $F=\MAX(E,E)$.
If $E=\OP(E_1,E_2)$, then we recursively construct two expressions $F_1$ and $F_2$, such that $F_i$ is equivalent to $E_i$ and $F_1 = \MAX(G_1,\dots,G_r)$, $F_2 = \MAX(H_1,\dots,H_q)$ (where $H_i$ and $G_i$ are max-free expressions), and we return $F=\MAX_{i\in\RangeSet{1}{r},j\in\RangeSet{1}{q}}\{\OP(G_i,H_j)\}$.
\end{proof}
Hence, in the rest of the paper we will assume w.l.o.g that all the expressions are of normal form.
The next example shows that optimal strategies for mean-payoff expressions may require infinite memory.
\begin{examp}\label{examp:InfiniteMem}
Consider the game graph in Figure~\ref{fig:InfiniteMemory} and the expression $E=\MAX(\LimInfAvg_1,\LimInfAvg_2)$.
In this game graph there is only one vertex that is controlled by player~1 and two self-loop edges, namely $e_1$ with $w(e_1) = (9,1)$ and $e_2$ with $w(e_2)=(1,9)$.
We first observe that any finite-memory strategy gives a value of at least $5$ to $E$.
Indeed, a finite-memory strategy induces an ultimately periodic path $\pi$ with $\LimInfAvg(\pi) = \alpha w(e_1) + (1-\alpha)w(e_2)$ for some $\alpha\in[0,1]\cap \Q$.
Hence, $E(\pi) = \MAX( 9\alpha + 1 - \alpha, \alpha + 9 - 9\alpha) = \MAX(8\alpha + 1, 9 - 8\alpha)$, and the minimum value for $E$ is obtained when $\alpha = \frac{1}{2}$ and we get that the minimal value for $E$ is $5$.
We now describe a player-1 infinite-memory strategy that gives a value of at most $2$ to $E$.
The strategy is simple. It follows $e_2$ as long as the average weight in the first dimension is more than $2$, then it follows $e_1$ as long as the average weight in the second dimension is more than $2$, and this process is repeated forever (i.e., $e_2$ is followed for a while, then $e_1$ and so on).
Clearly, in the formed path $\pi$ the average weight of the first dimension is at most $2$ for infinitely many prefixes of $\pi$.
Hence $\LimInfAvg_1(\pi) \leq 2$, and by the same arguments $\LimInfAvg_2(\pi) \leq 2$.
Thus, $E(\pi)\leq 2$, and we establish the fact that optimal strategies may require infinite-memory strategies (in this example, the presented infinite-memory strategy is not optimal, but we demonstrated that the best finite-memory strategy does not give an optimal value).
\begin{figure}
\begin{center} 
  \begin{gpicture}[name=loop](20,10)
	\node(P)(10,0){}
	\drawloop[loopangle=180](P){$(9,1)$}
	\drawloop[loopangle=0](P){$(1,9)$}
  \end{gpicture}
\end{center}
 \caption{}\label{fig:InfiniteMemory}
\end{figure}
\end{examp}
\subsection{Synthesis of a finite-memory controller for mean-payoff expression objectives}\label{subsec:MPExpGames}
In this subsection we apply Theorem~\ref{thm:AnalysisOfQuantitativeObj} to mean-payoff expression objectives.
We first prove that the solution for mean-payoff expressions satisfies Properties~1-3 , and thus the quantitative analysis problem is computable for mean-payoff expression games.
We then show that the boolean analysis problem is inter-reducible with Hilbert's tenth problem over rationals ($\HQ$) by showing that an effective algorithm for $\HQ$ implies that mean-payoff expressions satisfy Property~4, and by a reduction from $\HQ$ to mean-payoff expression games.

One-player games were solved in~\cite{mean-payoff-Automaton-Expressions} and in~\cite{YaronMPExp}.
We present our solution from~\cite{YaronMPExp} to establish properties of the one-player solution.
For an expression $E$ and a one-player game $(G,v_0)$, that is, a game over graph $G$ with initial vertex $v_0$, we say that a threshold $\nu$ is \emph{feasible} if the player has a strategy that achieves a value at least $\nu$ (we recall that in the one-player setting, the player aim to maximize the value of the objective).
The \emph{max-free} constraints were presented in~\cite{YaronMPExpFull} (Section A.4), and they describe the feasible thresholds of a max-free expression (a threshold $\nu$ is feasible if the one-player can achieve a value of at least $\nu$).
\begin{defn}[Max-free constraints]\label{def:maxFreeCon}
Let $G$ be a strongly-connected $k$-dimensional game graph,
and we recall that $C(G)$ is the set of simple cycles of $G$.
Let $E$ be a max-free expression such that the first $j$ dimensions of $G$ occur in $E$ as lim-inf (and the others as lim-sup).
We define a variable $X^i_c$ for every simple cycle $c$ and index $i\in\RangeSet{j+1}{k}$, and we define a vector of variables $\VEC{r}=(r_1,\dots,r_{2k})$.
Then the max-free constraints for threshold $\nu\in\Q$ are
\begin{enumerate}
\item $\sum_{c\in C(G)} X_c^i \Avg_m(c) \geq r_m
\mbox{  for every $i\in\Set{j+1,\dots,k}$ }\\ \mbox{and $m\in\Set{1,\dots,j,i}$}$
\item $\sum_{c\in C(G)} X_c^i = 1 \mbox{ for every $i\in\Set{j+1,\dots,k}$ }$
\item $X_c^i\geq 0 \mbox{ for every $i\in\Set{j+1,\dots,k}$ and $c\in C(G)$}$
\item $M_E \times \VEC{r} \geq (0,\dots,0,\nu)^T$
\end{enumerate}
\end{defn}
where $M_E$ is a matrix that is independent of the graph, and computable from $E$.
(We note that in~\cite{YaronMPExpFull}, the first type of constraints was $\sum_{c\in C(G)} X_c^i w_m(c)\geq r_m$, where $w_m$ is the projection of $w$ to the $m$-th dimension, and the second type of constraints was $\sum_{c\in C(G)}|c|X_c^i = 1$.
It is straight forward to verify that the constraints are equivalent --- in terms of feasibility.
In addition, the fourth constraint was presented as $M_E\times\VEC{r}\geq \VEC{b}_\nu$; but the proof of Lemma 7 in~\cite{YaronMPExpFull} implies that $\VEC{b}_\nu = (0,\dots,0,\nu)^T$.)
We proved in~\cite{YaronMPExp} that a threshold $\nu$ is feasible if and only if the corresponding max-free constraints are feasible.
For a max-free expression $E$, a strongly-connected graph $G$ and a threshold $\nu$, we denote the max-free constraints by $\MaxFreeConst(E,G,\nu)$ and we observe that for a (normal-form) mean-payoff expression $E=\MAX(E_1,\dots,E_\ell)$ and a strongly-connected graph $G$, the solution function for the one-player game is
$f(G)= \max \{\nu \in \R \mid \exists i\in\RangeSet{1}{\ell} \mbox{ s.t }\MaxFreeConst(E_i,G,\nu)\mbox{ is feasible}\}$.
By the definition of the max-free constraints, it easily follows that the solution is a function that is first-order definable and continuous (i.e., it satisfies Properties~1 and 3).
In the next Lemma we prove that the solution also satisfies the second property.
\begin{lem}\label{lem:DownWardCloseOfMP}
Let $E$ be a mean-payoff expression over $k$ dimensions, and let $f$ be its one-player solution function.
Then for every two strongly-connected graphs $G$ and $H$: if $\CONV(H)\subseteq \CONV(G)$, then $f(H) \leq f(G)$.
\end{lem}
\begin{proof}
Since we assume that $E=\MAX(E_1,\dots,E_n)$, where $E_i$ is a max-free expression, it is enough to prove that if a threshold $\nu$ is feasible in $H$ for the max-free expression $E_i$, then it is also feasible in $G$.
Let $c_1^G,\dots,c_n^G$ and $c_1^H,\dots,c_m^H$ be the simple cycles of $G$ and $H$ respectively.
We note that since $\CONV(H)\subseteq \CONV(G)$, then for every convex combination $x_1,\dots,x_m$, there is a convex combination $y_1,\dots,y_n$ such that $\sum_{i=1}^m x_i\Avg(c_i^H) = \sum_{i=1}^n y_i \Avg(c_i^G)$.
Hence, a solution for the max-free constraints over graph $H$ induces a solution for the max-free constraints over $G$ (by replacing, in the inequalities of constraints~1 over graph $H$, every convex combination of cycles of $H$ by the corresponding convex combination of cycles of $G$).

Thus, every threshold that is feasible for $H$ is also feasible for $G$, and the proof follows.
\pfbox
\end{proof}
Hence, the one-player solution function of mean-payoff expressions satisfies Properties~1-3 and the next theorem follows.
\begin{thm}\label{thm:MPExpOptimalComputable}
The quantitative analysis problem for mean-payoff expression games (where player 1 is restricted to finite-memory strategies) is computable.
\end{thm}
We now show that the solution for one-player mean-payoff expression games satisfies Property~4 if and only if $\HQ$ is decidable.
We first prove the direction from right to left.
\begin{lem}\label{lem:IfHQThenUltra}
If $\HQ$ is decidable, then mean-payoff expressions satisfy the fourth property.
\end{lem}
\begin{proof}
Let $G$ be an arbitrary strongly connected graph with $n$ simple cycles, let $C(G) = \{c_1,\dots,c_n\}$ be its set of simple cycles,
let $E=\MAX(E_1,\dots,E_m)$ be a mean-payoff expression (where $E_i$ is a max-free expression), and let $\nu$ be a rational threshold.
We recall that the sentence $(\zeta_n(\Avg(c_1),\dots,\Avg(c_n)\leq \nu)$ is equivalent to the statement:
"For every $y>\nu$ and $i\in\{1,\dots,m\}$, the constraints $\MaxFreeConst(E_i,G,y)$ are infeasible."
By the definition of the max-free constraints, when the set $\Avg(G)$ is fixed the above statement is easily reduced to the infeasibility of $m$ linear systems.
Motzkin's Transposition Theorem~(e.g., Theorem~1 in~\cite{ben2001motzkin}) gives a witness to the infeasibility of a set of linear inequalities.
We use Lagrange four-square Theorem to construct a Diophantine equation that has a rational root if and only if the witness exists.
We show that the construction works also when $\Avg(G)$ is not fixed, i.e., when $\Avg(G) = \{A_1\VEC{x_1},\dots,A_1\VEC{x_n}\}$ for some $n$ matrices $A_1,\dots,A_n$ and $n$ vectors of rational variables $\VEC{x_1},\dots,\VEC{x_n}\in \RatSimplex$.
(The details of the construction are given in the appendix.)
Hence, if $\HQ$ is decidable, then mean-payoff expressions satisfy the fourth property.
\pfbox
\end{proof}
We now prove the reduction from $\HQ$ to the boolean analysis of mean-payoff expression games, and we show that there is a reduction even for a simpler subclass of mean-payoff expressions.
An expression $E$ is $\SUM$-free and $\LimInfAvg$-only if only the $\MIN$ and $\MAX$ operators occur in $E$ and all the atomic expressions in $E$ are of the form $\LimInfAvg_i$. (In addition, the numerical complement operator also does not occur).
The next lemma shows that the boolean analysis problem for $\SUM$-free $\LimInfAvg$-only expressions is $\HQ$-hard.
\begin{lem}\label{lem:SumFreeInfOnlyHQ}
If the boolean analysis problem is decidable for $\SUM$-free $\LimInfAvg$-only expressions, then $\HQ$ is decidable.
\end{lem}
\begin{proof}
We only present a rough and informal sketch of the proof.
The full proof is given in the appendix.
We first show a reduction from $\HQ$ to the problem of finding a rational solution for two set of variables $Q=\{q_1,\dots,q_n\}$ and $P=\{p_1,\dots,p_n\}$ and a set of constraints, each of them is of the form:
(i)~$\sum_{i\in I} \alpha_i q_i \leq 0$, for some $I\subseteq \{1,\dots,n\}$; or
(ii)~$\sum_{i\in I} \alpha_i p_i \leq 0$, for some $I\subseteq \{1,\dots,n\}$; or
(iii)~$q_i p_j = q_k p_\ell$ for some $i,j,k,\ell\in \{1,\dots,n\}$; subject to $q_i, p_i > 0$.
We then show a reduction from the boolean analysis problem to the above problem.
We illustrate the reduction by showing the construction for the set of constraints $\{q_1 - 2q_2 \leq 0, 2p_1 - 3p_2 \leq 0, p_1q_1 = p_2q_2\}$.
For the these constraints we build a game graph $G$ that is illustrated in Figure~\ref{fig:SmallReduction}.
In the figure we explicitly show only part of the weight vectors and only part of the dimensions.
The initial vertex of $G$ is $s_0$ and this vertex is a player-2 vertex (and the rest are player-1 vertices).
The objective of the game is the expression
$E = \MAX(1,\MIN(2,3),\MIN(4,5))$, where $i$ stands for $\LimInfAvg_i$ (for $i=1,2,3,4,5)$, and the threshold is $0$.
In $G$, player 2 has only two memoryless strategies, namely $\tau_1 = s_0\to b_1$ and $\tau_1 = s_0\to a_1$.
We rely on Lemma~\ref{lem:WinningRegion} and show that player~1 has a winning strategy if and only if there is a finite path $\pi_1$ that visits only the self-loops of $b_2$ and a path $\pi_2$ that visits only the self-loops of $a_2$ such that every vector in $v\in\CONV(\Avg(\pi_1),\Avg(\pi_2))$ satisfies the winning condition, i.e., if $v=(v_1,v_2,v_3,v_4,v_5)$, then $\MAX(v_1,\MIN(v_2,v_3),\MIN(v_4,v_5))\leq 0$.
We observe that for any such path $\pi_1$ it holds that $\Avg(\pi_1) = q_1(3,1,0,-1,0) + q_2(-2,0,-1,0,1)$ and similarly $\Avg(\pi_2) = p_1(2,-1,0,1,0) + p_2(-1,0,1,0,-1)$ for some positive rational $q_1,q_2,p_1,p_2$.
We further observe that if $q_1 - 2q_2 > 0$, then $\Avg(\pi_1)$ is positive in the first dimension, and thus there is a vector $v\in\CONV(\Avg(\pi_1),\Avg(\pi_2))$ that gives a positive value to the expression.
Hence, it must hold that $q_1 -2q_2\leq 0$ and similarly $2p_1 - 3p_2 \leq 0$.
Moreover, we prove that if $p_1 q_1 \neq p_2 q_2$, then there exists $v\in\CONV(\Avg(\pi_1),\Avg(\pi_2))$ that is positive either in dimensions $2$ and $3$ or in dimensions $4$ and $5$.
Hence, it must hold that $p_1 q_1 = p_2 q_2$ and the proof follows.
 
\begin{figure}
\begin{center}
  \begin{gpicture}[name=loop](60,25)(0,20)
	\node[Nw=5,Nh=5,Nmr=0](P2)(30,30){$s_0$}

	\node[Nw=5,Nh=5](Alph1)(5,22.75){$a_1$}
	\node[Nw=5,Nh=5](Alph2)(55,22.75){$a_2$}

	\node[Nw=5,Nh=5](Bet1)(55,37.25){$b_1$}
	\node[Nw=5,Nh=5](Bet2)(5,37.25){$b_2$}

	\drawedge[ELside=r,curvedepth=0](P2,Alph1){}
	\drawedge[curvedepth=0](Alph1,Alph2){}
	\drawedge[ELside=r,curvedepth=0](Alph2,P2){}

	\drawedge[curvedepth=0](P2,Bet1){}
	\drawedge[curvedepth=0](Bet1,Bet2){}
	\drawedge[curvedepth=0](Bet2,P2){}

	\drawloop[ELpos=20,loopangle=180, loopdiam=2](Alph1){}
	\drawloop[ELpos=80,loopangle=0, loopdiam=2](Bet1){}

   \drawloop[loopangle=90, loopdiam=2](Alph2){\small{(2,-1,0,1,0)}}
   \drawloop[loopangle=270, loopdiam=2](Alph2){\small{(-1,0,1,0,-1)}}

   \drawloop[loopangle=90, loopdiam=2](Bet2){\small{(3,1,0,-1,0)}}
   \drawloop[loopangle=270, loopdiam=2](Bet2){\small{(-2,0,-1,0,1)}}
  \end{gpicture}
\end{center}
 \caption{}\label{fig:SmallReduction}
\end{figure}
\end{proof}
The next theorem summarizes the results of Lemmas~\ref{lem:SumFreeInfOnlyHQ} and~\ref{lem:IfHQThenUltra}
\begin{thm}\label{thm:HQandMPExp}
The boolean analysis problem for mean-payoff expression games (when player 1 is restricted to finite-memory strategies) is inter-reducible with $\HQ$, and it is $\HQ$-hard even for $\SUM$-free $\LimInfAvg$-only expressions.
\end{thm}
We also consider the case where both players are restricted to finite-memory strategies.
In this setting, the quantitative analysis problem is to compute
$\inf_{\sigma\in\FM_1}\sup_{\tau\in\FM_2} \Val_{\sigma,\tau}$.
The boolean analysis problem is to determine whether player 1 has a finite-memory strategy that assures a value of at most $\nu$ against any player-2 finite-memory strategy.
\begin{thm}\label{thm:BothFinite}
When both players are restricted to finite-memory strategies:
(i)~the quantitative analysis problem for mean-payoff expression games is computable; (ii)~the boolean analysis problem for mean-payoff expression games is inter-reducible with $\HQ$, and it is $\HQ$-hard even for $\SUM$-free $\LimInfAvg$-only expressions.
\end{thm}
\begin{proof}
We present a sketch of the proof. The full proof is given in the appendix.
Informally, when both players are restricted to finite-memory strategies, the outcome of a play is an ultimately periodic path, and thus we may assume that all the atomic expressions are of the form of $\LimInfAvg$ (because for periodic paths we have $\LimInfAvg(\pi) = \LimSupAvg(\pi)$).
We also show that if all the atomic expressions are $\LimInfAvg$ and player 1 is restricted to finite-memory strategies, then player 2 can achieve a value greater than $\nu$ if and only if he can do it with a finite-memory strategy, and the proof follows.\pfbox
\end{proof}
As a final remark, we note that while the boolean analysis for $\SUM$-free $\LimInfAvg$-only expressions is $\HQ$-hard when player 1 is restricted to a finite-memory strategy (and also when both players are restricted to finite-memory strategies), the next lemma shows that the problem is decidable when both players may use arbitrary strategies.
\begin{lem}[Theorem~5 in~\cite{VelnerR11}]\label{lem:BothArb}
When both players may use arbitrary strategies, the boolean analysis of $\SUM$-free $\LimInfAvg$-only expression games is decidable.
\end{lem}
\begin{proof}
The proof follows from Theorem~5 in~\cite{VelnerR11} due to the fact that there is an immediate translation from $\SUM$-free $\LimInfAvg$-only expressions to the $\bigvee\bigwedge \MPInfLeq{\nu}$ objectives that were defined in~\cite{VelnerR11}. \pfbox
\end{proof}

\section{Discussion and Future Work}
In this work we studied the synthesis of finite-memory strategies for games with robust multidimensional mean-payoff objectives, and we obtained two main results.
The first is a positive result, namely, the computability of the quantitative analysis problem.
The second has a negative flavour, and it shows that the boolean analysis is inter-reducible with Hilbert's Tenth problem over rationals.
From a practical point of view, the positive result is the most interesting, since for the first time (to the best of our knowledge) a recipe is given for computing $\epsilon$-optimal finite-memory strategies for a robust class of quantitative objectives.
A future work is to investigate whether the construction of these $\epsilon$-optimal strategies is feasible, both in terms of memory size and computational complexity.
From the theoretical point of view, the negative result is a bit surprising since it suggests that the boolean analysis is harder than the optimization problem, and in computer science typically there is a naive reduction from optimization problems to the corresponding decision problems.
However, in our case, the optimization computes only the greatest upper bound, and since optimal finite-memory strategies need not exist, then the reduction fails.
In fact, the hardness result suggests that it is even $\HQ$-hard to determine whether an optimal strategy exists.
A future work is to investigate games in which player~1 may use arbitrary infinite-memory strategies.
Additional direction for future work is to consider more general algebraic structures over multidimensional mean-payoff objectives.

\Heading{Acknowledgements.}
I would like to thank Prof. Alexander Rabinovich, my Ph.d supervisor, for many discussions on this work, and in particular for (i)~discussions on Hilbert tenth problem and Tarski's Theorem; and (ii)~suggesting to present a modular and abstract solution.

\newpage
\appendix
\section{Proof of Lemma~\ref{lem:IfHQThenUltra}}
Let $G$ be an arbitrary strongly connected graph with $n$ simple cycles, let $C(G) = \{C_1,\dots,C_n\}$ be its set of simple cycles,
let $E=\MAX(E_1,\dots,E_m)$ be a mean-payoff expression (where $E_i$ is a max-free expression), and let $\nu$ be a rational threshold.
We recall that the sentence $(\zeta_n(\Avg(C_1),\dots,\Avg(C_n)\leq \nu))$ is equivalent to the statement:
\begin{quote}
For every $y>\nu$ and $i\in\{1,\dots,m\}$, the constraints $\MaxFreeConst(E_i,G,y)$ are infeasible.
\end{quote}
By the definition of the max-free constraints, when the set $\Avg(G)$ is fixed the above statement is easily reduced to the infeasibility of $m$ linear systems, each of them is of the form:
\[A^i_{\Avg(G)}\VEC{x}\leq \VEC{b^i}\mbox{ and } B^i_{\Avg(G)}\VEC{x} < \VEC{c^i}\]
By Motzkin's Transposition Theorem~(e.g., Theorem~1 in~\cite{ben2001motzkin}) the infeasibility of a linear system
$A^i_{\Avg(G)}\VEC{x}\leq \VEC{b^i}\mbox{ and } B^i_{\Avg(G)}\VEC{x} < \VEC{c^i}$
is equivalent to the existence of two non-negative vectors $\VEC{y},\VEC{z}\geq \VEC{0}$ such that either 
\begin{itemize}
\item $\VEC{z}=0$ and $(A^i_{\Avg(G)})^T\VEC{y}=0$ and $\VEC{b^i}^T\VEC{y} < 0$; or
\item $\VEC{z}\neq 0$ and $(A^i_{\Avg(G)})^T\VEC{y} + (B^i_{\Avg(G)})^T\VEC{z}=0$ and $\VEC{b^i}^T\VEC{y} + \VEC{c^i}^T\VEC{z} \leq 0$
\end{itemize}
Since every linear inequality has a rational solution (when the coefficients are rational) we get that if such $\VEC{y}$ and $\VEC{z}$ exist, then there also exist rational $\VEC{y}$ and $\VEC{z}$ that satisfy the above.
Hence the above statement is equivalent to the rational feasibility of the following constraints (for variables $\VEC{y}=(y_1,\dots,y_r),\VEC{z}=(z_1,\dots,z_r),p_1,p_2$ and $q$):
\begin{itemize}
\item $p_1>0, p_2>0,q\geq 0$
\item $\VEC{y}\geq \VEC{0},\VEC{z}\geq \VEC{0}$
\item $(A^i_{\Avg(G)})^T\VEC{y} + (B^i_{\Avg(G)})^T\VEC{z}=0$
\item $(\sum_{j=1}^r z_i - p_1)(\VEC{b^i}^T\VEC{y} + p_2) = 0$
\item $(\sum_{j=1}^r z_i)(\VEC{b^i}^T\VEC{y} + \VEC{c^i}^T\VEC{z} + q) = 0$
\end{itemize}
By Lagrange's four-square Theorem, every natural number is the sum of four integer squares.
Therefore, every inequality of the form $x\geq 0$ is equivalent to the rational feasibility of the equation
\[x=\frac{x_1^2 + x_2^2 + x_3^2 + x_4^2}{1 + x_5^2 + x_6^2 + x_7^2 + x_8^2}\]
and every inequality of the form $x>0$ is equivalent to the rational feasibility of the equation
\[x=\frac{1+ x_1^2 + x_2^2 + x_3^2 + x_4^2}{1 + x_5^2 + x_6^2 + x_7^2 + x_8^2}\]
and the equations of the above form can be easily transformed into Diophantine equations.
Hence, we get that the infeasibility of a linear system
\[A^i_{\Avg(G)}\VEC{x}\leq \VEC{b^i}\mbox{ and } B^i_{\Avg(G)}\VEC{x} < \VEC{c^i}\]
is equivalent to the rational feasibility of several Diophantine equations $D_1=0,D_2=0,\dots,D_r=0$, and therefore it is equivalent to the rational feasibility of $D^i = D_1^2 + \dots + D_r^2 = 0$.
Therefore, when the set of simple cycles is fixed, the simultaneous infeasibility of all the max-free constraints is equivalent to the rational feasibility of the Diophantine equation $D_{\Avg(G)} = \sum_{i=1}^m (D^i)^2 = 0$.
We also note that if $\Avg(G)$ is not fixed, that is $\Avg(C_i)$ is a vector of variables (for $i=1,\dots,n$), then $D_{\Avg(C)}=0$ remains a Diophantine equation.

We are now ready to prove that the solution for one-player mean-payoff games is satisfies Property~4 (if $\HQ$ is decidable).
For $n$ matrices $A_1,\dots,A_n$ the rational satisfiability of $\zeta_n(A_1\VEC{x_1},\dots,A_n\VEC{x_n}) \leq \nu$ is equivalent to the existence of a rational solution to $D_{\Avg(G)}=0$ (for $\Avg(G)=\{A_1\VEC{x_1},\dots,A_n\VEC{x_n}\}$).
We can encode the requirement that $x_1,\dots,x_n\in\RatSimplex$ by a Diophantine equations $K(\VEC{x_1},\dots,\VEC{x_n})=0$ by the same techniques we used for the construction of $D_{\Avg(G)}=0$.
Hence, the satisfiability of $\zeta_n(A_1\VEC{x_1},\dots,A_n\VEC{x_n}) \leq \nu$ is equivalent to the existence of a rational solution to the Diophantine equation $K^2 + D^2 = 0$, and if $\HQ$ is decidable, then we can effectively determine whether $K^2 + D^2 = 0$ has a rational solution and the proof follows. \pfbox
\section{Proof of Lemma~\ref{lem:SumFreeInfOnlyHQ}}\label{sec:Hadness}
We prove Lemma~\ref{lem:SumFreeInfOnlyHQ} in the next three subsections.
In the first subsection we present an alternative formulation for $\HQ$.
In the second subsection we prove a simple technical lemma on vectors.
In the third subsection we present a reduction from the problem we presented in the first subsection to $\SUM$-free $\LimInfAvg$-only games, and the reduction relies on the lemma that we prove in the second subsection.

We note that the first two subsection are technical and tedious, but they relay only on basic algebra. 
\subsection{Alternative formulations of $\HQ$}
In this subsection, we present five problems; the first problem is $\HQ$, and we show a reduction from the $i$-th problem to the $i+1$-th problem, for $i=1,2,3,4$.
Thus, we get that there is a reduction from $\HQ$ to the fifth problem (that is, Problem~\ref{prob:H10QNormalForm4}), and in the third subsection we will show a reduction from that problem to mean-payoff expression games.
\begin{prob}[$\HQ$]\label{prob:H10Q}
For a polynomial $P$, find a rational solution to 
\[P(q_1,\dots,q_n) =0\]
\end{prob}
\begin{prob}\label{prob:H10QNormalForm1}
Find a rational solution to 
\[q_0 \cdot P(\frac{q_1}{q_0},\dots,\frac{q_n}{q_0}) = 0\] (for a polynomial $P$)
subject to
\begin{itemize}
\item $q_0 \leq q_i$ for every $i = 1,\dots,n$; and
\item $q_i \geq 1$ for every $i = 0,\dots,n$.
\end{itemize}
\end{prob}
\begin{lem}\label{lem:H10QNormalForm1}
There is a reduction from $\HQ$ to Problem~\ref{prob:H10QNormalForm1}.
\end{lem}
\begin{proof}
We first note that we can easily reduce $\HQ$ to the problem of finding a rational solution for the polynomial equation $D(q_1,q_2,\dots,q_n) = 0$ subject to $q_1,q_2,\dots,q_n \geq 1$. (The reduction is trivial, a polynomial equation $P(q_1,\dots,q_n)=0$ has a solution if and only if the polynomial equation $D(p_1,\dots,p_{2n}) = P(p_1 - p_2,p_3 - p_4, \dots, p_{2n - 1} - p_{2n}) = 0$ has a solution that satisfies $p_i\geq 1$.)
We define $P = D$ and we note that $q_0\cdot P = 0$ has a rational solution (subject to $q_0 \geq 1$) if and only if $P=0$ has a rational solution, and it is trivial to observe that $P=0$ has a rational solution (subject to $q_0\geq q_i$ and $q_i\geq 1$) if and only if $D=0$ has a rational solution (subject to $q_i\geq 1$).
\pfbox
\end{proof}

\begin{prob}\label{prob:H10QNormalForm2}
For a given a set of variables $Q=\{q_1,\dots,q_n\}$,
and a set of equations such that at most one equation is of the form
\begin{quote}
$\sum_{i\in I} \alpha_i q_i = 0$, for some $I\subseteq\{0,\dots,n\}$ and $\alpha_i\in\Q$ for every $i\in I$.
\end{quote}
and all the other equations are of the form
\begin{quote}
$q_i q_j = q_k q_\ell$ for some $i,j,k,\ell \in \{0,\dots,n\}$.
\end{quote}
find a rational solution that satisfies $1\leq q_0\leq q_i$ for every $i=1,\dots,n$.
\end{prob}
\begin{lem}\label{lem:H10QNormalForm2}
There is a reduction from Problem~\ref{prob:H10QNormalForm1} to Problem~\ref{prob:H10QNormalForm2}.
\end{lem}
\begin{proof}
We prove the lemma by giving a generic example that demonstrates the reduction.
Suppose that the equation with the form of Problem~\ref{prob:H10QNormalForm1} is
$q_0 \cdot P(\frac{q_1}{q_0},\dots,\frac{q_n}{q_0}) = 5\frac{q_1^2q_2q_3^3}{q_0^5} + \frac{q_1^2}{q_0} + 7q_0$,
then we reduce it to a problem with the form of Problem~\ref{prob:H10QNormalForm2} by defining the following equations:
\begin{itemize}
\item $p_0 \cdot q_0 = q_0 \cdot q_0$ (equivalent to $p_0 = q_0$)
\item $p_1 \cdot q_0 = q_1 \cdot q_1$ (equivalent to $p_1 = \frac{q_1^2}{q_0}$)
\item $p_2 \cdot q_0 = q_3\cdot q_3$ and $p_3 \cdot q_0 = p_2 \cdot q_3$ (equivalent to $p_3 = \frac{q_3^3}{q_0^2})$
\item $p_4 \cdot q_0 = p_1 \cdot p_3$ (equivalent to $p_4 = \frac{q_1^2 q_3^3}{q_0^4}$)
\item $p_5 \cdot q_0 = p_4 q_2$ (equivalent to $p_5 = \frac{q_1^2 q_2q_3^3}{q_0^5}$)
\item $5p_5 + p_1 + 7q_0 = 0$, subject to $1\leq q_0\leq q_1,q_2,q_3,p_0,p_1,p_2,p_3,p_4,p_5$ and $q_i,p_j \geq 1$ (equivalent to $q_0 \cdot P(\frac{q_1}{q_0},\dots,\frac{q_n}{q_0}) = 0$)
\end{itemize}
A solution to the above equations that satisfies $1\leq q_0\leq q_1,q_2,q_3,p_1,p_2,p_3,p_4,p_5$ is clearly a solution for $q_0\cdot P = 0$ 
that satisfies Problem~\ref{prob:H10QNormalForm1} conditions.
Conversely, a solution to $q_0\cdot P = 0$ that satisfies Problem~\ref{prob:H10QNormalForm1} conditions is a solution for the above constraints, and since $1\leq q_0 \leq q_1,q_2,q_3$ we also get that $q_0\leq p_1,p_2,p_3,p_4,p_5$ and a solution to the above equitations follows.
\pfbox
\end{proof}

\begin{prob}\label{prob:H10QNormalForm3}
For a given sets of variables $Q=\{q_1,\dots,q_n\}$, $P=\{p_1,\dots,p_n\}$, and a given set of equations, each of the form of either:
\begin{itemize}
\item $\sum_{i\in I} \alpha_i q_i = 0$, for some $I\subseteq \{1,\dots,n\}$; or
\item $q_i p_j = q_k p_\ell$ for some $i,j,k,\ell\in \{1,\dots,n\}$; or
\item $q_i = \frac{1}{2} \sum_{j=1}^n q_j$, for some $i\in \{1,\dots,n\}$; or
\item $p_i = \frac{1}{2} \sum_{j=1}^n p_j$, for some $i\in \{1,\dots,n\}$; or
\end{itemize}
find a rational solution that satisfies
\begin{itemize}
\item $q_1\leq q_i$, for $i=1,\dots,n$; and
\item $q_i,p_i\geq 1$ for $i=1,\dots,n$; and
\item $\sum_{i=1}^n p_i = \sum_{i=1}^n q_i$.
\end{itemize}
\end{prob}
\begin{lem}\label{lem:H10QNormalForm3}
There is a reduction from Problem~\ref{prob:H10QNormalForm2} to Problem~\ref{prob:H10QNormalForm3}.
\end{lem}
\begin{proof}
To show a reduction, we need to show how to encode an equation of the form of $q_1 q_2 = q_3 q_4$ with equations of the above form.
For this purpose we define the equations:
\begin{itemize}
\item $q_{n+1} = \frac{1}{2} \sum_{j=1}^{n+1} q_j$ and $p_{n+1} = \frac{1}{2} \sum_{j=1}^{n+1} p_j$
\item $q_2 p_{n+1} = q_{n+1} p_1$
\item $q_4 p_{n+1} = q_{n+1} p_2$
\item $q_1 p_1 = q_3 p_2$
\end{itemize}
It is straight forward to observe that if $\sum_{j=1}^{n+1} q_j = \sum_{j=1}^{n+1} p_j$ then the above set of equations are equivalent to $q_1q_2 = q_3q_4$.
\pfbox
\end{proof}

\begin{prob}\label{prob:H10QNormalForm4}
For a given sets of variables $Q=\{q_1,\dots,q_n\}$, $P=\{p_1,\dots,p_n\}$, and a given set of constraints, each of the form of either:
\begin{itemize}
\item $\sum_{i\in I} \alpha_i q_i \leq 0$, for some $I\subseteq \{1,\dots,n\}$; or
\item $\sum_{i\in I} \alpha_i p_i \leq 0$, for some $I\subseteq \{1,\dots,n\}$; or
\item $q_i p_j = q_k p_\ell$ for some $i,j,k,\ell\in \{1,\dots,n\}$
\end{itemize}
find a rational solution that satisfies
\begin{itemize}
\item $q_i,p_i > 0$ for $i=1,\dots,n$
\end{itemize}
\end{prob}
\begin{lem}\label{lem:H10QNormalForm4}
There is a reduction from Problem~\ref{prob:H10QNormalForm3} to Problem~\ref{prob:H10QNormalForm4}.
\end{lem}
\begin{proof}
The reduction is straight forward.
We replace every equation of the form of $\sum_{i\in I} \alpha_i q_i = 0$ with two constraints $\sum_{i\in I} \alpha_i q_i \leq 0$ and $\sum_{i\in I} -\alpha_i q_i \leq 0$.
We replace $q_i = \frac{1}{2} \sum_{j=1}^n q_j$ with $\sum_{j\in\{1,\dots,n\}-\{i\}} \frac{1}{2} q_j - \frac{1}{2}q_i \leq 0$ and
$\sum_{j\in\{1,\dots,n\}-\{i\}} -\frac{1}{2} q_j + \frac{1}{2}q_i \leq 0$.
We replace $p_i = \frac{1}{2} \sum_{j=1}^n p_j$ with $\sum_{j\in\{1,\dots,n\}-\{i\}} \frac{1}{2} p_j - \frac{1}{2}p_i \leq 0$ and
$\sum_{j\in\{1,\dots,n\}-\{i\}} -\frac{1}{2} p_j + \frac{1}{2}p_i \leq 0$.
In addition, we add $n$ constraints $q_1\leq q_i$ for $i=1,\dots,n$.
It is straight forward to observe that if the above formed constraints have a rational solution $Q=\{q_1,\dots,q_n\},P=\{p_1,\dots,p_n\}$ that satisfies $q_i,p_i > 0$, then for every rational positive $m$ we get that $mQ=\{mq_1,\dots,mq_n\},P=\{p_1,\dots,p_n\}$ and  $Q=\{q_1,\dots,q_n\},mP=\{mp_1,\dots,mp_n\}$ are also solutions.
Hence, a solution to the formed constraints implies that there is a solution that satisfies $q_i,p_i \geq 1$ and $\sum_{i=1}^n p_i = \sum_{i=1}^n q_i$.
And conversely, if the formed constraints are not satisfiable, then clearly the original equations are not solvable. \pfbox
\end{proof}

\subsection{Auxiliary lemma}
In this subsection, we prove the next lemma.
\begin{lem}\label{lem:LemmaForReduction}
Let $\alpha_1,\alpha_2,\beta_1,\beta_2$ be strictly positive rationals, and
let
$v_1(\alpha_1) = \alpha_1 \cdot (-1,0,1,0)$, 
$v_2(\alpha_2) = \alpha_2 \cdot (0,1,0,-1)$,
$u_1(\beta_1) = \beta_1 \cdot (1,0,-1,0)$, and
$u_2(\beta_2) = \beta_2 \cdot (0,-1,0,1)$.
For every $m,n\in\Q$ we denote by the vector $x(m,n)=(x_1,x_2,x_3,x_4)$ the sum $m(v_1+v_2) + n(u_1 + u_2)$.
Then the following assertions are equivalent:
\begin{enumerate}
\item $\frac{\beta_1}{\alpha_1} = \frac{\beta_2}{\alpha_2}$.
\item For every non-negative rationals $m,n$: 
$\MAX(\MIN(x_1,x_2),\MIN(x_3,x_4)) \leq 0$.
\end{enumerate}
\end{lem}
\begin{proof}
By definition $x_1 = -m\alpha_1 + n\beta_1$, $x_2 = m\alpha_2 -n\beta_2$,
$x_3 = -x_1$ and $x_4 = -x_2$.

We first prove that assertion~1 implies assertion~2.
Suppose that $\frac{\beta_1}{\alpha_1} = \frac{\beta_2}{\alpha_2}$,
let $m$ and $n$ be arbitrary non-negative rationals, and we denote $k=\frac{m}{n}$.
In order to prove that $\MAX(\MIN(x_1,x_2),\MIN(x_3,x_4)) \leq 0$, it is enough to show that if $x_1 > 0$, then $x_2 < 0$ (since in this case $x_3 = -x_1 < 0$).
Suppose that $x_1 > 0$. Hence, $\beta_1 > k\alpha_1$, and we get that $k<\frac{\beta_1}{\alpha_1}$.
By definition, $x_2 = n( k\alpha_2 - \beta_2)$, and since we assumed that $\frac{\beta_1}{\alpha_1} = \frac{\beta_2}{\alpha_2}$, and we proved that $k<\frac{\beta_1}{\alpha_1}$, we get that $x_2 < 0$, and the claim that assertion~1 implies assertion~2 follows.

In order to prove that assertion~2 implies assertion~1, we consider two distinct cases.
In the first case we assume (towards a contradiction) that $\frac{\beta_1}{\alpha_1} > \frac{\beta_2}{\alpha_2}$, and we choose $m$ and $n$ that satisfy
$\frac{\beta_1}{\alpha_1} > k=\frac{m}{n} > \frac{\beta_2}{\alpha_2}$.
We claim that $x_1 > 0$ and $x_2 > 0$, and therefore a contradiction to the assumption that $\MAX(\MIN(x_1,x_2),\MIN(x_3,x_4)) \leq 0$ follows.
Indeed, since $\frac{\beta_1}{\alpha_1} > k$, then $x_1 = n(-k\alpha_1 + \beta) >0$, and since $k > \frac{\beta_2}{\alpha_2}$, then $x_2 = n(k\alpha_2 - \beta_2) > 0$.
In the second case, we assume that $\frac{\beta_1}{\alpha_1} < \frac{\beta_2}{\alpha_2}$, and by similar arguments, we get that $x_3,x_4>0$ and a contradiction follows.
Hence, in both cases we get that assertion~2 implies assertion~1, and the proof of the lemma follows.
\pfbox
\end{proof}

\subsection{The reduction}
In this subsection, we present a reduction from Problem~\ref{prob:H10QNormalForm4} to the boolean synthesis problem for mean-payoff expressions (when player 1 is restricted to finite-memory strategies).
The reduction is as following:
For a given sets of variables $Q=\{q_1,\dots,q_n\}$, $P=\{p_1,\dots,p_n\}$, and a given set of constraints, each of the form of either:
\begin{itemize}
\item $\sum_{i\in I} \alpha_i q_i \leq 0$, for some $I\subseteq \{1,\dots,n\}$; or
\item $\sum_{i\in I} \alpha_i p_i \leq 0$, for some $I\subseteq \{1,\dots,n\}$; or
\item $q_i p_j = q_k p_\ell$ for some $i,j,k,\ell\in \{1,\dots,n\}$
\end{itemize}
We denote by $t_1$ the number of constraints that are of the first form, and w.l.o.g we assume that the number of constraints that are of the second form is also $t_1$.
We denote by $t_2$ the number of constraints that are of the third form.
We construct a $k = 2 + n + t_1+4t_2$ dimensional game graph with $5$ states (see Figure~\ref{fig:Reduction}), and an expression
\[E = \MAX(\LimInfAvg_1,\dots,\LimInfAvg_{2+n+t_1},E_1,\dots,E_{2t_2})\]
where
\[E_i = \MIN(\LimInfAvg_{2+n+t_1+2i},\LimInfAvg_{2+n+t_1+2i+1})\]
The transitions of the graph are described in Figure~\ref{fig:Reduction}, and each of the states $a_2$ and $b_2$ has $n$ self-loop edges.
\begin{figure}[H]
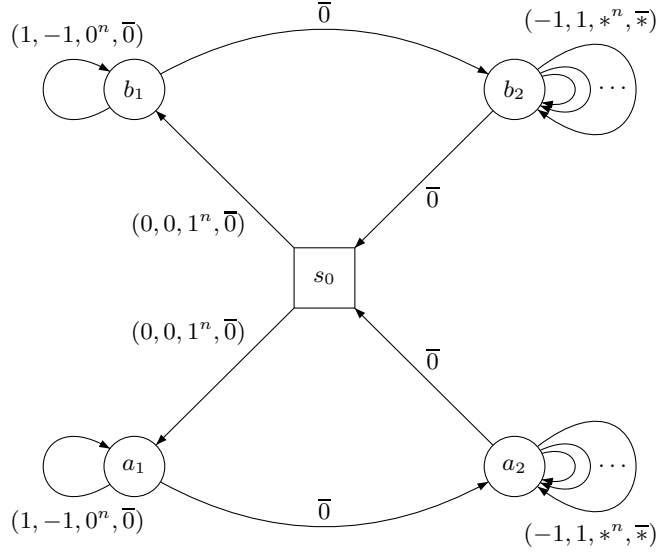

\begin{center} 
  \begin{gpicture}[name=loop](60,80)
	\node[Nmr=0](P2)(30,30){$s_0$}

	\node(Alph1)(5,5){$a_1$}
	\node(Alph2)(55,5){$a_2$}

	\node(Bet1)(5,55){$b_1$}
	\node(Bet2)(55,55){$b_2$}

	\drawedge[ELside=r,curvedepth=0](P2,Alph1){$(0,0,1^n,\VEC{0})$}
	\drawedge[curvedepth=-8](Alph1,Alph2){$\VEC{0}$}
	\drawedge[ELside=r,curvedepth=0](Alph2,P2){$\VEC{0}$}

	\drawedge[curvedepth=0](P2,Bet1){$(0,0,1^n,\VEC{0})$}
	\drawedge[curvedepth=8](Bet1,Bet2){$\VEC{0}$}
	\drawedge[curvedepth=0](Bet2,P2){$\VEC{0}$}

	\drawloop[ELpos=20,loopangle=180](Alph1){$(1,-1,0^n,\VEC{0})$}
	\drawloop[ELpos=80,loopangle=180](Bet1){$(1,-1,0^n,\VEC{0})$}

   \drawloop[loopdiam=4,loopangle=0](Alph2){}
   \drawloop[loopdiam=6,loopangle=0](Alph2){}
   \drawloop[ELpos=80,loopdiam=12,loopangle=0](Alph2){$(-1,1,*^n,\VEC{*})$}
	\put(66,5){$\dots$}

   \drawloop[loopdiam=4,loopangle=0](Bet2){}
   \drawloop[loopdiam=6,loopangle=0](Bet2){}
   \drawloop[ELpos=20,loopdiam=12,loopangle=0](Bet2){$(-1,1,*^n,\VEC{*})$}  
	\put(66,55){$\dots$}
  \end{gpicture}
\end{center}
 \caption{The graph that is formed by the reduction. $s_0$ is player-2 state and $a_1,a_2,b_1,b_2$ are player-1 states.
$\VEC{0}$ denotes a vector of zeros; $1^n$ denotes a vector of $n$ ones, $0^n$ denotes a vector of $n$ zeros, and the weights of $\VEC{*}$ and $*^n$ are given in the description of the reduction}\label{fig:Reduction}
\end{figure}
The weight vector $\VEC{w}$ of the $i$-th self-loop edge of state $a_2$ is determined according to the next rules:
\begin{enumerate}
\item The first two dimensions of $\VEC{w}$ are $-1$ and $+1$ (respectively).
Intuitively, this assures that player 1 will not stay forever in state $a_1$ or in state $a_2$.
\item The weight of dimension $2+i$ is $-1$ and for $j\in\{1,\dots,n\}-\{i\}$ the weight of dimensions $j$ is $0$.
Intuitively, this assures that player 1 will visit edge $i$ at least once.
\item If the $j$-th type-1 equation is $\sum_{m\in I} \alpha_m q_m \leq 0$, then if $i\in I$, then the weight in dimension $2 + n + j$ is $-\alpha_m$.
Otherwise, we assign zero for this dimension.
Intuitively, this enforce player 1 to visits edge $i$ for $q_i$ times in such way that $\sum_{m\in I} \alpha_m q_m \leq 0$.
\item If the $j$-th type-3 equation is $q_m p_r = q_k p_\ell$, then the weights of the four dimensions
$2+n+t_1+4j,2+n+t_1+ 4j+1,2+n+t_1+ 4j+2,2+n+t_1+ 4j+3$ are:
\begin{itemize}
\item If $i=m$, then the weights are $(-1,0,1,0)$
\item If $i=k$, then the weights are $(0,1,0,-1)$
\item Otherwise, the weights are $(0,0,0,0)$
\end{itemize}
\end{enumerate}
The weight vector $\VEC{w}$ of the $i$-th self-loop edge of state $b_2$ is determined according to the next rules:
\begin{enumerate}
\item The first $2+n$ dimensions are determined by the same rules that we presented to the self-loop edges of state $a_2$.
\item If the $j$-th type-2 equation is $\sum_{m\in I} \alpha_m p_m \leq 0$, then if $i\in I$, then the weight in dimension $2 + n + j$ is $-\alpha_m$.
Otherwise, we assign zero for this dimension.
Intuitively, this enforce player 1 to visits edge $i$ for $p_i$ times in such way that $\sum_{m\in I} \alpha_m p_m \leq 0$.
\item If the $j$-th type-3 equation is $q_m p_r = q_k p_\ell$, then the weights of the four dimensions
$2+n+t_1+ 4j, 2+n+t_1+ 4j+1,2+n+t_1+ 4j+2,2+n+t_1+ 4j+3$ are:
\begin{itemize}
\item If $i=m$, then the weights are $(1,0,-1,0)$
\item If $i=k$, then the weights are $(0,-1,0,1)$
\item Otherwise, the weights are $(0,0,0,0)$
\end{itemize}
\end{enumerate}
In the rest of this subsection, we will prove that player 1 has a finite-memory strategy that assures non-positive value for the expression $E$ if and only if the given set of equations has a solution that satisfies Problem~\ref{prob:H10QNormalForm4} limitations.

In the next lemmas we prove key properties of the game.
The first lemma characterized the one-player game solution for the expression $E$.
\begin{lem}\label{lem:OnePlayerSolForReduction}
Let $G$ be an arbitrary strongly connected $k$-dimensional weighted one-player game graph, and let $f$ be the one-player solution for the expression $E$.
Then $f(G) > 0$ if and only if
\begin{itemize}
\item $G$ has a simple cycle with positive average weight in a dimension $i\in \{1,\dots,2+n+t_1\}$; or
\item $G$ has two simple cycles $C_1$ and $C_2$, and there exist an index $i\in\{1,\dots,2t_2\}$ and two positive rationals $m,n$ for which
\begin{quote}
$m\Avg(C_1)+n\Avg(C_2)$ is positive is dimension $2+n+t_1+2i$ and in dimension $2+n+t_1+2i+1$.
\end{quote}
\end{itemize}
\end{lem}
\begin{proof}
The proof follows directly by the definitions of the max-free constraints (Definition~\ref{def:maxFreeCon}). \pfbox
\end{proof}
\begin{lem}\label{lem:WinsInReductionIffWinsInSat}
In the mean-payoff expression game over game graph $G$ (that is constructed by the reduction) and threshold $0$, player 1 wins from vertex $s_0$ if and only if he has a finite-memory strategy $\sigma$ such that $f(\CONV(G^\sigma)) \leq 0$ (where $f$ is the one-player solution for the expression $E$).
\end{lem}
\begin{proof}
By the construction of $G$ it follows that if player 1 strategy is to loop for ever in state $a_1$ or $b_1$, then the lim-inf of the average weight in dimension 1 will be 1 and $E$ will get a positive value.
Similarly, if player 1 strategy is to loop forever in state $b_2$ or $a_2$, then the average weight in dimension 2 is positive, and so does the value of $E$.
Hence, every player-1 winning strategy will visit the initial state $s_0$ infinitely often.
Therefore, if $\sigma '$ is a player-1 winning strategy, then every SCC in $G^{\sigma '}$ contains a vertex $(s_0, m)$ (for some memory state $m$).
Let $S$ be a terminal SCC in $G^{\sigma '}$ and let $(s_0,m)$ be a vertex in $S$.
We construct the witness strategy $\sigma$ by changing the initial memory state o $\sigma '$ to $m$.
If $\sigma '$ is a winning strategy, then by definition $f(S) \leq 0$ and since $G^\sigma = S$ we get that $f(\CONV(G))\leq 0$.

Hence, if player 1 wins in the game, then such $\sigma$ exists, and the proof for the converse direction is trivial (since such a strategy $\sigma$ is a winning strategy).
\pfbox
\end{proof}
In the game graph $G$ player 2 has only two possible memoryless strategies: the first strategy is to follow the edge $(s_0,a_1)$, and we denote this strategy by $\tau_1$, and the second strategy is to follow $(s_0,b_1)$, and we denote it by $\tau_2$.
\begin{lem}\label{lem:NewCondition}
There exists a player-1 strategy for which $f(G^\sigma) \leq 0$ if and only if
there exist cyclic paths $\pi_1$ and $\pi_2$ such that $\pi_i$ is a cyclic path in $G^{\tau_i}$ that visits all the edges of $G^{\tau_i}$ and $f(\CONV(\pi_1,\pi_2)) \leq 0$.
\end{lem}
\begin{proof}
By Lemma~\ref{lem:ConvCyclesMemoryless} such $\sigma$ exists if and only if
there exist two ultimately periodic paths $\rho_1$ and $\rho_2$ such that
$\rho_i$ is an infinite path in the graph $G^{\tau_i}$ and $f(\CONV(\Avg(\rho_1),\Avg(\rho_2))) \leq 0$.
Hence, the proof for the direction from right to left follows.
In order to prove the converse direction we assume that such $\rho_1$ and $\rho_2$ exist and show how to construct $\pi_1$ and $\pi_2$.
Let $\rho_1 = \pi_0 (\pi_1)^\omega$ (i.e., $\pi_1$ is the periodic finite path in $\rho_1$).
We claim that if $\pi_1$ does not contain all the edges of $G^{\tau_1}$, then $f(\{\Avg(\pi_1)\}) > 0$.
The proof of the claim is by considering the following distinct cases:
\begin{itemize}
\item Case 1: if $\pi_1$ contains only the cycles $s_0\to a_1 \to a_2\to s_0$, then the value of $\Avg(\pi)$ is positive in the third dimension.
\item Case 2: if $\pi_1$ contains only the self loop of $a_1$, then the value of the first dimension of $\Avg(\pi)$ is positive
\item Case 3: if $\pi_1$ does not contain the self loop of $a_1$, and contains some of the self loops of $a_2$, then the second dimension of $\Avg(\pi)$ is positive.
\item Case 4: if $\pi_1$ contains the cycle $s_0\to a_1 \to a_2\to s_0$, the self loop of $a_1$ and not the $i$-th self loop of $a_2$, then dimension $2+i$ of $\Avg(\pi)$ is positive.
\end{itemize}
Hence, if $\pi_1$ does not contain all the edges of $G^{\tau_1}$, then we get that $f(\Avg(\rho_1)) > 0$ (since $\Avg(\rho_1) = \Avg(\pi_1)$), and since $f$ is monotone, we get that $f(\CONV(\Avg(\rho_1), \Avg(\rho_2)) > f(\Avg(\rho_1)) > 0$, which contradict the definition of $\rho_1$.
We construct the witness path $\pi_2$ in a similar way (i.e., by defining $\rho_2 = \pi_0 ' (\pi_2)^\omega$, and the proof that $\pi_2$ contains all the edges of $G^{\tau_2}$ is similar.
Since $\Avg(\pi_i) = \Avg(\rho_i)$, we get that $f(\CONV(\Avg(\pi_1),\Avg(\pi_2)) \leq 0$ and the proof is complete.
\pfbox
\end{proof}
We now give two additional definitions and then prove the correctness of the reduction.
Let $C_1,\dots,C_n$ be the simple cycles of $G^{\tau_1}$. We denote $\RatSimplex(G^{\tau_1}) = \{\{v\in\Q^k\mid \exists (x_1,\dots,x_n)\in\RatSimplex(n)\mbox{ s.t } v = \sum_{i=1}^n x_i \Avg(C_i)\}$,
and we similarly define $\RatSimplex(G^{\tau_2})$.
We say that two vectors $v_1$ and $v_2$ are \emph{satisfactory} if $f(\CONV(v_1,v_2))\leq 0$.
We are now ready to prove the correctness of the reduction, and by Lemma~\ref{lem:WinsInReductionIffWinsInSat} and Lemma~\ref{lem:NewCondition} it is enough to prove that there exists $v_i\in\RatSimplex(G^{\tau_i})$ (for $i=1,2$) such that $v_1,v_2$ are satisfactory vectors if and only if the given set of equations has a rational solution.

We first prove the direction from right to left.
Suppose that the given set of equations has a rational solution $P,Q$ that satisfies $q_i,p_i >0$.
We construct the vector $v_1\in\RatSimplex(G^{\tau_1})$ by taking $\frac{1}{1+2\sum_{i=1}^n q_i}$ fraction of the average weight of the cycle $s_0\to a_1\to a_2\to s_0$, $\frac{2\sum_{i=1}^n q_i}{1+2\sum_{i=1}^n q_i}$ fraction of the average weight of the self loop of $a_1$ and
$\frac{q_i}{1+2\sum_{i=1}^n q_i}$ fraction of the average weight of the $i$-th self loop of $a_2$.
Similarly, we construct the vector $v_2\in\RatSimplex(G^{\tau_2})$ by taking $\frac{1}{1+2\sum_{i=1}^n p_i}$ fraction of the average weight of the cycle $s_0\to b_1\to b_2\to s_0$, $\frac{2\sum_{i=1}^n p_i}{1+2\sum_{i=1}^n p_i}$ fraction of the average weight of the self loop of $b_1$ and
$\frac{p_i}{1+2\sum_{i=1}^n p_i}$ fraction of the average weight of the $i$-th self loop of $b_2$.
By the construction of $G$, and since $P$ and $Q$ are solutions for the equations, it is straight forward to verify that the first $2+n+t_1$ dimensions of $v_1$ and $v_2$ are non-positive.
In addition, by Lemma~\ref{lem:LemmaForReduction}, and since $P$ and $Q$ satisfies all the equations of the form $q_i p_j = q_k p_\ell$, we get that for every positive $m,n\in\Q$ we have that $m v_1 + n v_2$ are non-positive in dimension $2+n+t_1+2i$ or in dimension $2+n+t_1+2i+1$ for every $i=1,\dots,2t_4$.
Hence, by Lemma~\ref{lem:OnePlayerSolForReduction}, the vectors $v_1,v_2$ are satisfactory.

Conversely, suppose that there exist $v_i\in\RatSimplex(G^{\tau_i})$ (for $i=1,2$) such that $v_1,v_2$ are satisfactory vectors.
We denote by $w_{s_0,a}$ the average weight of the cycle $s_0\to a_1 \to a_2 \to s_0$, by $w_{a_1}$ the average weight of the self loop of $a_1$, and by $w_{a_2}^i$ the average weight of the $i$-th self loop of $a_2$.
By definition, there exists $n+2$ positive rationals $x,y,q_1,\dots,q_n$ for which $v_1 = xw_{s_0,a} + yw_{a_1} + \sum_{i=1}^n q_i w_{a_2}^i$.
Similarly, we denote by $w_{s_0,b}$ the average weight of the cycle $s_0\to b_1 \to b_2 \to s_0$, by $w_{b_1}$ the average weight of the self loop of $b_1$, and by $w_{b_2}^i$ the average weight of the $i$-th self loop of $b_2$, and
by definition, there exists $n+2$ positive rationals $x,y,p_1,\dots,p_n$ for which $v_1 = xw_{s_0,b} + yw_{b_1} + \sum_{i=1}^n p_i w_{b_2}^i$.
We claim the $Q=\{q_1,\dots,q_n\},P=\{p_1,\dots,p_n\}$ are a solution to the given set of equations.
By Lemma~\ref{lem:OnePlayerSolForReduction} and by the construction of the graph, it immediately follows that $Q$ and $P$ satisfy all the type-1 and type-2 constraints.
In addition, by Lemma~\ref{lem:LemmaForReduction} (and by Lemma~\ref{lem:OnePlayerSolForReduction}) we get that all the type-3 equations are also satisfied.
Hence, we get that if there exist $v_i\in\RatSimplex(G^{\tau_i})$ (for $i=1,2$) such that $v_1,v_2$ are satisfactory vectors, then the given set of constraints have a solution.

To conclude, we get that the boolean analysis problem for mean-payoff expressions is harder than $\HQ$, and the proof of Lemma~\ref{lem:SumFreeInfOnlyHQ} follows. \pfbox
\section{Proof of Theorem~\ref{thm:BothFinite}}
When both players are restricted to finite-memory strategies the outcome of the game is an ultimately periodic path $\pi = \pi_1 (\pi_2)^\omega$.
Thus, for every dimension $i$ we have $\LimInfAvg_i(\pi) = \LimSupAvg_i(\pi)$.
Hence, w.l.o.g we may assume that the game objective is a $\LimInfAvg$-only expression.
In this section, we will show a reduction from games in which both players are restricted to finite-memory strategies to games in which only player 1 is restricted to finite-memory strategies.
The reduction is based on the next lemma.
\begin{lem}\label{lem:FromInfToFin}
Let $E$ be a $\LimInfAvg$-only expression and let $G$ be a multidimensional weighted graph, and the goal of player 1 is to assure $E\leq \nu$.
Then a player-1 finite-memory strategy is winning if and only if it wins against every player-2 finite-memory strategy.
\end{lem}
\begin{proof}
The proof for the direction from left to right is trivial.
To prove the converse direction we fix a player-1 finite-memory strategy $\sigma$ and we show that if player 2 a strategy that wins against $\sigma$, then he also has a finite-memory winning strategy.
We note that when $\sigma$ is fixed, a player-2 strategy is an infinite path in $G^\sigma$ and a player-2 finite-memory strategy is an ultimately periodic path in $G^\sigma$.
Hence, there exists an infinite path $\pi$ in $G^\sigma$ for which $E$ assigns a value greater than $\nu$.
We claim that for every $\epsilon> 0$ there is an ultimately periodic path $\rho$ in $G^\sigma$ such that in every dimension $\LimInfAvg_i(\rho_\epsilon) \geq \LimInfAvg_i(\pi) - \epsilon$.
Indeed, let $s$ be a state that is visited infinitely often by $\pi$, and let $\pi_s$ be a suffix of $\pi$ that begins in state $s$, and we observe that $\LimInfAvg(\pi_s) = \LimInfAvg(\pi)$.
By the definition of $\LimInfAvg$ and by the finiteness of the graph it follows that for every $\epsilon > 0$ there exists a path $\pi_\epsilon$ that is a prefix of $\pi_s$, ends in state $s$, and $\LimInfAvg_i(\pi_\epsilon) \geq \LimInfAvg_i(\pi_s) - \epsilon$.
We denote by $\pi_0$ the shortest path from the initial state to $s$, and we get that the ultimately periodic path $\rho_\epsilon = \pi_0 (\pi_\epsilon)^\omega$ satisfies the assertion of the claim.
To complete the proof of the lemma, we denote the number of $\SUM$ operators in $E$ by $\#\SUM$ and we set $\epsilon = \frac{E(\pi) - \nu}{2\#\SUM}$.
It is easy to verify that the ultimately periodic path $\rho_\epsilon$ satisfies $E(\rho_\epsilon) \geq E(\pi) - \frac{E(\pi)-\nu}{2} = \frac{E(\pi) + \nu}{2} > \nu$, and the proof follows. \pfbox
\end{proof}
The proof of Theorem~\ref{thm:BothFinite} follows immediately from the fact that we only consider $\LimInfAvg$-only expressions and from  Lemma~\ref{lem:FromInfToFin} and Theorems~\ref{thm:MPExpOptimalComputable} and~\ref{thm:HQandMPExp}.
\end{document}